%% file: main.tex
\documentclass[conference]{IEEEtran}
\IEEEoverridecommandlockouts
\usepackage[T1]{fontenc}

\usepackage{cite}
\usepackage{amsmath,amssymb,amsfonts}
\usepackage{algorithmic}
\usepackage{graphicx}
\usepackage{textcomp}
\usepackage{xcolor}
\usepackage{hyperref} 
\usepackage{braket}
\usepackage{dsfont}
\usepackage{subcaption}
\usepackage{cite}
\usepackage{amsmath,amssymb,amsfonts}
\usepackage{algorithm}
\usepackage{algorithmic}
\usepackage{textcomp}
\usepackage{mathtools,amsthm}

\def\BibTeX{{\rm B\kern-.05em{\sc i\kern-.025em b}\kern-.08em
    T\kern-.1667em\lower.7ex\hbox{E}\kern-.125emX}}
       \newtheorem{thm}{Theorem}
\def\BibTeX{{\rm B\kern-.05em{\sc i\kern-.025em b}\kern-.08em
    T\kern-.1667em\lower.7ex\hbox{E}\kern-.125emX}}
    
\def\BibTeX{{\rm B\kern-.05em{\sc i\kern-.025em b}\kern-.08em
    T\kern-.1667em\lower.7ex\hbox{E}\kern-.125emX}}
\begin{document}
%
\title{SCOOP: A Quantum-Computing Framework for Constrained Combinatorial Optimization\thanks{This research was supported in part by a National Sciences and Engineering Research Council (NSERC) of Canada Collaborative Research and Training Experience (CREATE) grant on Quantum Computing, NSERC Alliance Consortium Grant entitled Quantum Software Consortium -- Exploring Distributed Quantum Solutions for Canada (QSC), and NSERC Alliance grant on Quantum Computing for Optimal Mobility. 
\\Corresponding author: pangara@uvic.ca}}

\author{\IEEEauthorblockN{Prashanti Priya Angara}
\IEEEauthorblockA{
\textit{University of Victoria}\\
Victoria, BC, Canada \\
pangara@uvic.ca} 
\and
\IEEEauthorblockN{Emily Martins}
\IEEEauthorblockA{
\textit{University of Victoria}\\
Victoria, BC, Canada \\
emilymartins@uvic.ca}
\and
\IEEEauthorblockN{Ulrike Stege}
\IEEEauthorblockA{
\textit{University of Victoria}\\
Victoria, BC, Canada \\
ustege@uvic.ca} 
\and
\IEEEauthorblockN{Hausi M\"uller}
\IEEEauthorblockA{
\textit{University of Victoria}\\
Victoria, BC, Canada \\
hausi@uvic.ca} 
}

\maketitle

\input{sections/00-abstract}
\input{sections/01-introduction}

\input{sections/02-background}
\input{sections/03-profit-framework}

\input{sections/dominating-set}

\input{sections/min-max-matching}
\input{sections/set-cover}
\input{sections/experimental-setup}
\input{sections/results}

\input{sections/conclusions}
\newpage
\bibliographystyle{IEEEtran}
\bibliography{IEEEabrv, hubos}

\end{document}

%% file: sections/00-abstract.tex
\begin{abstract}
While the ultimate goal of solving computationally intractable problems is to find a provably optimal solutions, practical constraints of real-world scenarios often necessitate focusing on efficiently obtaining high-quality, near-optimal solutions. The Quantum Approximate Optimization Algorithm (QAOA) is a state-of-the-art hybrid quantum-classical approach for tackling these challenging problems that are encoded using quadratic and higher-order unconstrained binary optimization problems (QUBO and HUBO). We present SCOOP, a novel QAOA-based framework for solving \textit{constrained} optimization problems. SCOOP transforms a constrained problem into an unconstrained  counterpart, forming SCOOP problem \textit{twins}. The QAOA quantum algorithm operates on the unconstrained twin to identify potential optimal and near-optimal solutions. Effective classical post-processing reduces the solution set to the constrained problem space. Our SCOOP approach is solution-enhanced, objective-function-compatible, and scalable. We demonstrate the framework on three NP-hard problems, \textsc{Minimum Dominating Set}, \textsc{Minimum Maximal Matching}, and \textsc{Minimum Set Cover} appearing in practical application domains such as resource allocation, communication networks, and machine learning. We validate SCOOP's feasibility and effectiveness on Xanadu PennyLane simulators.

\end{abstract}

%% file: sections/01-introduction.tex
\section{Introduction}
\label{sec:introduction}

\vspace*{-3pt}


When solving optimization problems, we often seek methods that not only determine optimal solutions but also identify near-optimal solutions. In applications where trade-offs between accuracy, efficiency, and practicality are paramount, identifying near-optimal solutions can hold significant value over solely seeking optimal solutions. We also seek methods that are scalable; one approach that hampers scalability is problem-specific tuning (e.g., penalty parameter tuning). 

Quantum computing offers promising avenues to tackle optimization problems, especially classically intractable combinatorial optimization problems (COPs).  Both quantum annealing (QA)~\cite{johnson2011quantum} and the quantum approximate optimization algorithm (QAOA)~\cite{farhi2014quantum} are well-suited techniques. 
However, to apply these methods, the original COP must be transformed to be encoded on a quantum computer. 

The QAOA is a variational quantum algorithm specifically designed to solve \textit{unconstrained} COPs on gate-based quantum hardware. 
The Ising model~\cite{johnson2011quantum, Lucas2014} is widely used in quantum combinatorial optimization because it provides a natural mathematical structure for encoding binary decision problems in a way that maps directly onto quantum hardware and algorithms like  QA or QAOA. The field of quantum optimization relies heavily on Ising model formulations that use quadratic terms due to their direct compatibility with QA hardware. Gate-based hardware has no such restriction, as higher-order terms can be encoded natively using multi-qubit gates~\cite{cowtan2019phase}. The higher order unconstrained binary optimization (HUBO) formulation offers a more expressive way to represent many optimization problems by directly incorporating multi-way ($K$-local) interactions between variables~\cite{wang_speedup_2025}. Rapid advancements in gate-based quantum computers, particularly in qubit fidelity and connectivity, are paving the way for more effective handling of multi-qubit interactions~\cite{ibmExpandsQuantum}, and consequently the encoding of HUBO problems. 

\begin{table*}[hbt!]
\centering
\caption{\textbf{SCOOP twins}: NP-hard constrained  COPs that allow the derivation of an unconstrained transformational equivalent COPs. \textbf{Legend}: \textsc{MinVC}: Minimum Vertex Cover, \textsc{MaxIS}: Maximum Independent Set, \textsc{MaxCl}: Maximum Clique} 
\begin{tabular}{|l|l|l|l|}
\hline
\textbf{Constrained Twin} & \textbf{Unconstrained Twin} & \textbf{Constraint Avoided} & \textbf{Classical Post-processing} \\
\hline
\textsc{MinDS} [Sec.~\ref{sec:ds}] & \textsc{MaxPD} (HUBO) & Vertex domination & Add vertices [Alg.~\ref{alg:add_vertices}]  \\
\hline
\textit{\textsc{M$^3$}} [Sec.~\ref{sec:min-max}] & \textsc{MaxPES} (HUBO)& Maximality, matching & Add edges [Alg.~\ref{alg:add_edges}], Ensure matching [Alg.~\ref{alg:eds_to_mm}]  \\
\hline
\textsc{MinEDS} [Sec.~\ref{sec:min-max}] & \textsc{MaxPES} (HUBO)& Edge domination & Add edges [Alg.~\ref{alg:add_edges}]  \\
\hline
\textsc{MinIEDS} [Sec.~\ref{sec:min-max}] & \textsc{MaxPES} (HUBO)& Edge domination, edge independence & Add edges [Alg.~\ref{alg:add_edges}], Ensure independence of edges [Alg.~\ref{alg:eds_to_mm}]  \\
\hline
\textsc{MinSC} [Sec.~\ref{sec:setCover}] & \textsc{MaxPSC} (HUBO)& Set cover & Add elements [Alg.~\ref{alg:add_elements}]  \\
\hline
\textsc{MinVC} & \textsc{MaxPC} (QUBO) & Edge coverage & Add vertices~\cite{Stege2002, Angara2025}  \\
\hline
\textsc{MaxIS} & \textsc{MaxPI} (QUBO) & Edge coverage & Remove conflicted vertices~\cite{van2008tractable, Angara2025}  \\
\hline
\textsc{MaxCl} & \textsc{MaxPCl} (QUBO) & Complete subgraph & Remove conflicted vertices in complement graph~\cite{scott2004classical, Angara2025}  \\
\hline
\end{tabular}

\label{tab:scpp_twin_mapping}
\end{table*}


\textit{Constrained} COPs involve constraints that limit the set of feasible solutions, thereby restricting the solution space.  
To model constrained COPs, penalties are typically introduced to discourage infeasible solutions~\cite{glover2022quantum}. While this approach guarantees to retain optimal solutions, it can encourage non-optimal infeasible solutions, since choosing optimal penalty values is challenging~\cite{smith1997penalty, roch2023effect}. 

The field of quantum computing and optimization has seen significant advancements in recent years, with extensive studies on QUBOs and penalty methods to encode COPs. In 2014, Farhi introduced the QAOA technique~\cite{farhi2014quantum}, which is specifically designed for solving unconstrained problems on gate-based quantum computers. Meanwhile, QA has historically been limited to problems with quadratic terms, requiring higher-order terms to be reduced through a process called quadratization. 
In 2019,  Hadfield et al.~\cite{hadfield2019quantum} introduced an alternative method for constrained problems that restricts quantum evolution to the feasible subspace via carefully constructed quantum operators.



\color{black}


A third alternative, distinct from both penalty-based methods and feasibility-preserving approaches, was recently proposed by Angara et al.~\cite{Angara2025}. This work investigates transformational equivalent unconstrained ``profit'' variants of three transformational equivalent COPs. These unconstrained variants are solved using  Farhi's original QAOA formulation (or \textit{vanilla} QAOA), followed by classical post-processing. For these three problems, the authors demonstrate improved performance and practical scalability in terms of approximation ratios, as well as summed optimal and near-optimal probability metrics.

In this paper, we introduce the novel SCOOP framework to solve \textit{constrained} COPs using vanilla QAOA. {Table~\ref{tab:scpp_twin_mapping} presents the NP-hard constrained problems alongside their unconstrained counterparts. It also details the type of transformations used to render all solution subsets feasible, and refers to the corresponding classical polynomial-time post-processing algorithm required to map solutions for the unconstrained problem into the feasible subspace of the original, \textit{constrained}, COP.
This framework provides a better encoding that enables the discovery of optimal and near-optimal solutions at scale. To achieve this, SCOOP guides the derivation of an unconstrained COP from the constrained COP. This process results in a penalty-free cost Hamiltonian that is well-suited for QAOA. Subsequently, the viable solutions obtained by QAOA are efficiently post-processed. Our approach offers the advantages that it scales for all inputs and supports the determination of near-optimal solutions.


The main contributions of our paper are: (1) The SCOOP framework that guides the derivation of an unconstrained COP from a constrained COP. The SCOOP derivation leads to a penalty-free cost Hamiltonian that is palatable for QAOA. The viable solutions produced by QAOA are efficiently post-processed.
    (2) We formalize the conditions that the relationship between the constrained and unconstrained COPs must satisfy. 
    (3) Following the SCOOP framework, we derive unconstrained COPs to solve the constrained problems listed in Table~\ref{tab:scpp_twin_mapping}.
    (4) For each unconstrained COP, we present the cost Hamiltonian that can serve as input to QAOA.
    (5) For both \textit{\textsc{Minimum Dominating Set}} (\textsc{MinDS}) and \textit{\textsc{Minimum Maximal Matching}} (\textsc{M$^3$}), we discuss in detail the setup and  results of our experiments  performed on Xanadu's Pennylane simulator for 3-regular graphs with up to ten qubits and eight layers of QAOA.  To the best of our knowledge, experimental results for both \textsc{MinDS} and \textit{\textsc{M$^3$}} using \textit{quantum} techniques remain unexplored in the literature.

 In the remainder of this paper, we discuss related work (Sec.~\ref{sec:background}) and 
define the necessary relationship between a constrained and an unconstrained COP, and introduce our SCOOP framework (Sec.~\ref{sec:profit-framework}). 
We derive SCOOP twins and their Hamiltonians in Sections~\ref{sec:ds}--\ref{sec:setCover}.
Sections~\ref{sec:methodology} and~\ref{sec:results} present an evaluation of the SCOOP framework using Xanadu PennyLane, and Section~\ref{sec:conclusions} concludes the paper.  



%% file: sections/02-background.tex
\vspace*{-5pt}

\section{Related Work}\label{sec:background}

\vspace*{-3pt}

We  introduce the problems of study,  selected to illustrate our SCOOP framework, related work QAOA, as well as  on encoding optimization problems for the use with QAOA. 

\subsubsection*{Problems of Study}
In the early 1970s, among the first sets of problems proved to be NP-hard are \textit{\textsc{Minimum Dominating Set}},  \textit{\textsc{Minimum Set Cover}}, and \textit{\textsc{Minimum Maximal Matching}} \cite{garey1979computers}. We picked these problems for our case studies and give problem definitions in Secs.~\ref{sec:profit-framework}--\ref{sec:setCover}. Both \textit{\textsc{MinDS}} and \textit{\textsc{M$^3$}} remain hard even when the input graphs are restricted to bipartite graphs of maximum degree 3~\cite{chlebik2008approximation,yannakakis1980edge}. 
These problems have applications in various domains: Finding small dominating sets has applications in social networks \cite{Bouamama2021}, wireless networks \cite{Bai2020, Liang2023} and biological networks---for the identification of disease-genes and to extract biological motifs \cite{Nacher2016}. 
Minimum set covers find applications in networking, such as in the lifetime maximization of wireless sensor networks \cite{Zhang2016} and link-state routing approaches \cite{Mihailovic2022}.
Minimum maximal matchings are valuable in resource allocation  \cite{demange2008minimum}. 

\subsubsection*{Quantum Approximate Approximation Algorithm}
QAOA, introduced by Farhi et al.~in 2014~\cite{farhi2014quantum}, is a hybrid quantum-classical heuristic algorithm designed to solve combinatorial optimization problems. The standard \textit{vanilla} QAOA framework includes: (1) an initial equal superposition state, $\vert \psi_0 \rangle = \ket{+}^{\otimes n}$, (2) an ansatz comprising alternating parameterized cost unitaries ($U(\boldsymbol{\gamma})=e^{-i \boldsymbol{\gamma} \hat{H}_C}$) and mixing unitaries ($U(\boldsymbol{\beta}) = e^{-i \boldsymbol{\beta} \hat{H}_M}$) applied over $p$ layers, and (3) a classical optimizer that minimizes the expected value of the cost Hamiltonian. The algorithm’s performance depends on the choice of ansatz, parameters $(\boldsymbol{\gamma}, \boldsymbol{\beta} )$, and classical optimizer, with common methods including L-BFGS, COBYLA, Adam, and Gradient Descent. Numerous variants of QAOA targetting constrained optimization have emerged over the years, notably by Hadfield et al.~\cite{hadfield2019quantum}, Saleem et al.~\cite{SaleemTTS23, tomesh2023divide}, Egger et al.~\cite{egger2021warm}. 
\color{red}
\color{black}
In existing studies, QAOA performance has primarily been evaluated on Erdős-Rényi random graphs with varying edge probabilities~\cite{golden2023numerical, SaleemTTS23} and bounded-degree graphs, particularly random 3-regular graphs~\cite{shaydulin2023qaoawith, lykov2021performance}. Herrman et al.~\cite{herrman2021impact}  investigated \textit{\textsc{MaxCut}} performance using QAOA with up to three layers across various small graphs (up to eight nodes).

\subsubsection*{Unconstrained Binary Optimization Problems}
While QA strictly requires a QUBO formulation, QAOA generally can handle both QUBO and HUBO problems. Initially, lower-order terms and, consequently, QUBOs might be preferred. However, converting a HUBO into a QUBO through quadratization necessitates introducing auxillary variables and additional penalties \cite{glover2022quantum}. In contrast, Campbell et al.~\cite{campbell2022qaoa} investigate the advantages of incorporating higher-order terms (HOTs) into the problem formulation. Their research underscores the potential of utilizing the inherent capabilities of gate-model quantum computers to manage HOTs, suggesting an expanded application domain for QAOA and paving the way for near-term quantum superiority on NP-hard problems.  

Pelofske et al.~\cite{pelofske2023quantum} conducted an experimental comparison between QAOA and QA on different hardware platforms. Their study focused on cubic ZZZ interactions, which are natively constructible for QAOA, while QA necessitates order reduction to quadratic interactions using auxiliary variables. 

Mandal et al.~\cite{mandal2020compressed} study quadratization and propose a method of degree reduction that works directly in the Ising space. 
They note that a \emph{sparse} problem in the Ising space—i.e., its cost Hamiltonian, defined by a sum of local terms that each involves a small, bounded number of variables and each variable appears in only a small number of terms—\emph{is not necessarily sparse} in Boolean space and vice versa. 
A simple example illustrates this: Consider polynomial $2^n \prod_{i=1}^n b_i$. Applying the transformation from binary to Ising variables ($b_i \leftarrow (1 - Z_i)/2$) results in the Hamiltonian $\sum_{I \subseteq \{1, \dots, n\}} (-1)^{|I|} \prod_{i \in I} Z_i$, which grows exponentially \cite{mandal2020compressed}. 
In the context of the traveling salesperson problem (TSP), Glos et al.~\cite{Glos_Krawiec_Zimboras_2022} propose a trade-off between the number of qubits and the circuit depth to demonstrate an efficient binary encoding for TSP.

We note that the feasibility of the encoding of HUBOs or HOTs into a cost Hamiltonians is instance dependent; understanding the structural properties of the problem is crucial for assessing whether or not higher-order terms must be used as further discussed in Secs.~\ref{sec:methodology} and  \ref{sec:results}.

\vspace*{-6pt}

%% file: sections/03-profit-framework.tex
 

\vspace*{-5pt}

\section{The SCOOP framework}


\label{sec:profit-framework}
In this section, we formalize a blueprint that guides us in deriving an unconstrained COP, $P_U$, from a constrained one, $P_C$. $P_U$ possesses the necessary properties to obtain optimal and near-optimal solutions using the vanilla QAOA approach for $P_U$, followed by efficient post-processing.

In general terms, a combinatorial optimization problem $P$ takes as input, an instance $x$ from its input domain $\cal I$ (e.g., a graph from the set of all simple undirected graphs).  
For a given $x \in I$,  a solution is $s\in \text{Sol}(x)$, where $\text{Sol}(x)$ denotes the general set of possible solutions for $x$ (e.g.,  if the instance is a graph $G= (V,E)$ and every solution is a subset of vertices, then $\text{Sol}(G) = {\cal P}(V)$, the powerset of $V$). 
The objective function of $P$ for an instance $x\in I$ can be described as $\texttt{obj}^x: \text{Sol}(x)\longrightarrow \mathbb Q$ 
(e.g., if the measure of the solution is size, and a solution to a graph problem is a vertex subset, say $V'\subseteq V$, then $\texttt{obj}^x(s) = \texttt{obj}^G(V')= |V'|$). 
A typical format for describing a COP  $P$: For $x\in I$ and $s\in \text{Sol}(x)$, \linebreak

\vspace*{-25pt}

\begin{align*}
&\text{Optimize:} \quad |\texttt{obj}^x(s)| \\
\end{align*}

\vspace*{-15pt}

\noindent Note that in the case of an unconstrained COP, any $s\in \text{Sol}(x)$ is a feasible solution. When talking about a COP that is unconstrained, we may call it $P_U$. If instead the COP considered is a \textit{constrained} COP, then the set of feasible solutions is a subset of $\text{Sol}(x)$, denoted
$\text{Sol}_C(x)$ where $\text{Sol}_C(x)\subset \text{Sol}(x)$.

A typical format for describing a constrained COP  $P_C$ is: for $x\in I$ and $s\in \text{Sol}(x)$, 

\vspace*{-12pt}

\begin{align*}
&\text{Optimize:} \quad |\texttt{obj}^x(s)| \\
&\text{Subject to:} \quad s\in \text{Sol}_C(x)
\end{align*}



\color{black}

Before describing the process of how to derive an unconstrained COP $P_U$ from a given constrained COP $P_C$, we define a relationship between the problems that is sufficient to obtain optimal and near-optimal solutions using the vanilla QAOA approach for $P_U$, followed by efficient post-processing.

 We say that a constrained COP $ P_C $ with input domain $I_C$, feasible solutions $\text{Sol}_C$, and objective function $\texttt{obj}_C$, and an unconstrained COP $ P_U $ with input domain $I_U$, solutions $\text{Sol}_U = \text{Sol}$, and objective function $\texttt{obj}_U$,  are \textit{\textbf{s}olution-enhanceable, \textbf{c}onstrained-unconstrained,  \textbf{o}bjective-function-compatible   \textbf{o}ptimization \textbf{p}roblem twins} or \textit{SCOOP twins}
if the following conditions hold:


\begin{enumerate}
    \item \textbf{Identical Input Domains}: The problems share the same input domain $ I $, i.e., $I = I_C = I_U$.
    
    \item \textbf{Solution Containment}: for every instance $ x \in I $, $ \text{Sol}_C(x) \subseteq \text{Sol}_U(x) $. That is, each feasible solution  for $ P_C $ is a solution for $ P_U $. 
    
    \item \textbf{Objective Function Compatibility}: For any $x\in I$ there exists a 
%
   polynomial-time computable function 
$ f_x: \mathbb{Q} \to \mathbb{Q} $ such that for  all $ s \in \text{Sol}_C(x) $:
$
\texttt{obj}^x_C(s) = f_x(\texttt{obj}^x_U(s)) $  and $ \texttt{obj}^x_U(s) = f_x^{-1}(\texttt{obj}^x_C(s)).$

    \item \textbf{Solution Enhance-ability}: For each $x\in I$, and each $s\in \text{Sol}_U(x)$, $s\in \text{Sol}_C(x) $ or there exists $s'\in \text{Sol}_C(x)$, where $s'$ is computable in polynomial time from $s$. Furthermore,  $\texttt{obj}^x_U(s')$ is at least as good as  $\texttt{obj}^x_U(s)$ and $\texttt{obj}^x_C(s')$ is at least as good as  $f_x(\texttt{obj}^x_U(s))$.
\end{enumerate}

\noindent We next describe a process that may allow, given a constrained  problem $P_C$, to derive an unconstrained  problem $P_U$ such that $P_C$ and $P_U$ are SCOOP twins.
Let $P_C$ be a constrained combinatorial optimization problem where for $x\in I$, a feasible solution $s \in \text{Sol}_C(x)$ has cost $\texttt{obj}^x_{C}(s)$.

\begin{description}
    \item[\textbf{Step 1}]~Identify the constraints of $P_C$. 
    \item[\textbf{Step 2}]~Consider any infeasible solution $s'$ for $P_C$. Identify a way to quantify the extent to which the solution satisfies the constraint, deriving $\texttt{obj}^x_{U}(.)$ by relating it to   $\texttt{obj}^x_{C}(s')$.
    \item[\textbf{Step 3}]~Develop $P_U$ using cost function $\texttt{obj}^x_{U}(.)$.
    \item[\textbf{Step 4}]~Prove objective function compatibility.
    \item[\textbf{Step 5}]~Prove solution enhance-ability.
\end{description}

\smallskip

\noindent Note that the existence of the SCOOP twin $P_U$ for a given $P_C$ yields NP-hardness for $P_U$ as long as $P_C$ is NP-hard.


We illustrate our framework considering the NP-hard problem {\sc Minimum Maximal Matching (M$^3$)} \cite{garey1979computers}, which is defined as follows.
For a simple graph  undirected graph $G = (V,E)$ and a subset $\textit{M}\subseteq E$ of edges, 

\vspace*{-15pt}

\begin{align*}
&\text{Minimize:} \quad |M| \\
&\text{Subject to:}\quad\text{(1)~} M \text{~is a matching: for every pair of edges}\\
&\quad\quad\quad\quad\quad\quad ab,cd \in M, ab\not= cd, \quad a,b,c \text{~and~} d  \text{~are}\\
&\quad\quad\quad\quad\quad\quad \text{pairwise distinct}\\
& \text{~\!~\!}\quad\quad~\quad~\quad\text{~(2)~} \text{matching~} M \text{~is a maximal: there is no edge}\\
&\quad\quad\quad\quad\quad\quad e \in E\text{~with~} M\cup\{e\}  \text{~is a matching}\\
\end{align*}

\vspace*{-15pt}

\noindent Here, $\texttt{obj}^G_{M^3}(M) = |M|$. Solutions to the problem are subsets of edges (i.e., $\text{Sol}(G) = |E|$) that must satisfy constraints (1) and (2), describing $\text{Sol}_C(G) \subset \text{Sol}(G)$. To quantify how much a subset $E'$ of edges satisfies the constraints, when $E'$ is not feasible, that is $E'$ does not satisfy one or both of the constraints, we say that an edge $e\in E$ is \textit{covered} by $E'$ if there is an edge $e'\in E'$ such that $e$ and $e'$ share a common endpoint. In other words every, edges covered by $E'$ are all edges in $E'$ plus all edges in $E\setminus E'$ that are "adjacent" to $E'$.
Now we can define $\texttt{obj}^G_{\textit{PES}}(E') = |\{e\in E \mid e \text{~is covered by~} E'\}|-|E'|$.

We are ready define the problem {\textit{\textsc{Maximum Profitable Edge Set (MaxPES)}}:  for  a graph $G = (V,E)$ and $\textit{E'}\subseteq E$, 

\vspace*{-15pt}

\begin{align*}
&\text{Maximize:} \quad \text{~profit~} \mathfrak p_{\text{PES}~} \text{where~} \mathfrak p_{\text{PES}~} =|\texttt{obj}^G_{\textit{PES}}(E')| \\
\end{align*}

\vspace*{-15pt}

\noindent In Sec.~\ref{sec:min-max} we prove that {\textit{\textsc{M$^3$}} and  {\textit{\textsc{MaxPES}} are indeed SCOOP twins. 
In the following section, we document that our framework applies to the two NP-hard problems {\textit{\textsc{Minimum Dominating Set}} and {\textit{\textsc{Minimum Set Cover}} (Secs.~\ref{sec:ds},\ref{sec:setCover}).




%% file: sections/dominating-set.tex
\section{Minimum Dominating Set}
\label{sec:ds}

A \textit{dominating set} in an undirected graph $G = (V, E)$\footnote{In this section, we refer to simple undirected graphs without isolated vertices.} with vertex set $V$ and edge set $E$ is a subset $DS \in V$  such that every vertex in $G$ is ``dominated'' by $\textit{DS}$. A vertex is considered to dominate itself and all of its adjacent vertices. This ensures that all vertices in the graph are either part of the dominating set or directly connected to a vertex in it. A smallest possible dominating set is referred to as \textit{minimum dominating set}. 

We define the problem \textit{\textsc{Minimum Dominating Set}} (\textit{\textsc{MinDS}}) for a graph $G = (V,E)$ as follows.
A subset of vertices  $\textit{DS}\subseteq V$ is a \textit{dominating set} if every vertex  $v \in V$ is either in  $\textit{DS}$ or adjacent to at least one vertex in $\textit{DS}$. The objective is to find the smallest such set:
\vspace*{-5pt}
\begin{align*}
&\text{Minimize:} \quad |\textit{DS}| \\
&\text{Subject to:~} \text{~for all~} v \in V\\ 
& \quad\quad\ v \in \textit{DS~} \text{~or~} \text{~there exists~} 
 u \in N(v) \ \mbox{such that} \ u \in \textit{DS}
\end{align*}

\noindent Using our SCOOP framework for \textit{\textsc{MinDS}}, we derive the unconstrained problem \textit{\textsc{Maximum Profit Domination}} or \textsc{MaxPD}. Theorem \ref{thm:equiv-ds-pd} below concludes the proof that \textit{\textsc{MinDS}} and \textit{\textsc{MaxPD}} are SCOOP twins.  

\noindent \textbf{Step 1: Identify constraints} \\
The constraint that \textit{\textsc{MinDS}} must satisfy to ensure feasibility is that the subset $\textit{DS} \subseteq V$ is a dominating set.

\noindent \textbf{Step 2: Quantify feasibility of infeasible solution} \\
Consider the \textit{extent} to which a subset of vertices \textit{dominates vertices} and its \textit{closeness} to being a dominating set (i.e., the more vertices are dominated, the better).

\noindent \textbf{Step 3: Develop objective function $\texttt{obj}^G_{PD}$}\\
For a subset $\textit{PD}\subseteq V$ in $G$,\\
$
\texttt{obj}^G_{PD}(PD) = |\{v \in V| v~\text{is dominated by}~ PD\}| - |PD|
$. We also refer to the value produced by the objective function $\texttt{obj}^G_{PD}(PD)$ as \textit{profit}, $\mathfrak{p}_{PD}$.

\smallskip

 Before moving on to Step 4, we formulate the problem {\sc Maximum Profit Domination (MaxPD)}~\cite{fernau2019profit} for a graph $G=(V, E)$ and $\textit{PD} \subseteq V$: 
 \vspace{-5pt}
\begin{align*}
& \text{Maximize:} \text{~profit~} \mathfrak p_{\text{PD}}, \text{where~} 
{\mathfrak p}_{\text{PD}} = \texttt{obj}^G_{PD}(PD)
 \end{align*}

\vspace*{-5pt}

\noindent \textbf{Step 4: Objective function compatibility} \\
From Theorem~\ref{thm:equiv-ds-pd} below, we can derive function $f_G(.)$. 
That is, $f_G(\mathfrak{p}) = |V| - k$ and $f_G^{-1}(k) = |V| - \mathfrak{p}$.


\noindent \textbf{Step 5: Solution enhance-ability}\\
Alg.~\ref{alg:add_vertices} describes a polynomial-time post-processing procedure that converts a solution $\textit{PD}$ for \textit{\textsc{MaxPD}} to a domintating set, i.e., a solution for  \textit{\textsc{MinDS}} while preserving solution quality.

\color{black}



\begin{thm}~\cite{fernau2019profit}\label{thm:equiv-ds-pd}
For any graph $G =(V,E)$, $G$ has a dominating set $\textit{DS}\subseteq V$ of size $k$ if and only if $G$ has a subset $\textit{PD}\subseteq V$ with profit ${\mathfrak p}_{\text{PD}}=|V| - k$. 
\end{thm}

\begin{proof} (Sketch) \item $(\Rightarrow)$  Determining the profit ${\mathfrak p}_{\text{PD}}$ of a dominating set $\textit{DS} \subseteq V$ for $G$ results in ${\mathfrak p}_{\text{PD}}=\texttt{obj}^G_{PD}(DS)$. 
Since $\textit{DS}$ is a dominating set 
${\mathfrak p}_{\text{PD}}= |E| - k$.

\item $(\Leftarrow)$  Consider a subset $PD \subseteq V$ with profit $\mathfrak{p}_{\text{PD}} = |V| - k$. If  $\textit{PD}$ is a dominating set, then 
$|\textit{PD}| = k$ and thus $\textit{PD}$ is a dominating set of size $k$. 

If instead $\textit{PD}$ is not a dominating set then 
we can obtain a dominating set of size at most $|V| - {\mathfrak p}_{\text{PD}}$ by applying Algorithm~\ref{alg:add_vertices}. Note that the algorithm at no step when expanding $\textit{PD}$ reduces the profit.~\end{proof}





\vspace*{-10pt}
\begin{algorithm}
\caption{Converting solutions: \textit{\textsc{MaxPD}}  to \textit{\textsc{MinDS}} adapted from~\cite{fernau2019profit}}
\label{alg:add_vertices}
\begin{algorithmic}

\REQUIRE Graph $G$, Solution $\textit{PD}$
\STATE Let $V_{\text{nd}}$ be the set of non-dominated vertices in $G$
\FOR{each vertex $v \in V_{\text{nd}}$}
    \IF{$N(v)  \not\subseteq  \textit{PD}$ and $v \notin \textit{PD}$}
        \STATE Add $v$ to $\textit{PD}$
    \ENDIF
\ENDFOR
\end{algorithmic}
\end{algorithm}

\noindent Theorem~\ref{thm:equiv-ds-pd} implies that, given an optimal solution to \textsc{MaxPD} $\textit{PD}$ for a graph $G$, we can obtain a minimum dominating set for $G$ using Algorithm \ref{alg:add_vertices}.

Having obtained the unconstrained SCOOP twin \textsc{MaxPD} for \textsc{MinDS} we are ready to describe the formulations of the cost Hamiltonians. 

\subsection*{Cost Hamiltonians for \textit{\textsc{MinDS}} and \textit{\textsc{MaxPD}}} 
We describe cost Hamiltonians for both \textit{\textsc{MinDS}} ($\hat{H}_{DS}$) and \textit{\textsc{MaxPD}} ($\hat{H}_{PD}$) using higher-order unconstrained binary formulations. Let $\vec{x} = (x_1, x_2, \dots, x_{|V|})$ be a vector of binary variables, where $x_i \in \{0, 1\}$ for each $i \in V$, then,
\vspace*{-2pt}
\begin{align*}
    \hat{H}_{DS}(\vec{x}) &=  \hat{H}^{DS}_{N}(\vec{x}) + \hat{H}^{DS}_{V}(\vec{x}) \\
    &= A\sum_{i \in V} \Bigg[\Big((1-x_i)\prod_{j\in N(i)} (1-x_j)\Big)\Bigg]+B\sum_{v \in V} x_v,
\end{align*}
with, $A>B$, and
\vspace*{-2pt}
\begin{align*}
    \hat{H}_{PD}(\vec{x}) &=  \hat{H}^{PD}_{N}(\vec{x}) - \hat{H}^{PD}_{V}(\vec{x}) \\
    &=  \sum_{i \in V} \Bigg[1 - \Big((1-x_i)\prod_{j\in N(i)} (1-x_j)\Big)\Bigg] - \sum_{v \in V} x_v.
\end{align*}

We formulated the HUBO of the constrained COP \textit{\textsc{MinDS}} following the guidelines of Glover et al.~\cite{glover2022quantum}. The goal of this HUBO is to minimize $\hat{H}_{DS}(\vec{x})$. $\hat{H}_{DS}(\vec{x})$ imposes a penalty for vertices that are not dominated by a factor of $A$. Its minimum value is   $\hat{H}_{DS}(\vec{x}) = B \sum_{v \in V} x_v$. 

For the unconstrained COP \textsc{MaxPD}, the goal is to maximize the unconstrained objective $\hat{H}_{PD}(\vec{x})$. The term $\hat{H}^{PD}_{N}(\vec{x})$ keeps track of all vertices that are dominated---the contribution of each vertex is $1$ if dominated and $0$ otherwise.

We point out that Dinneen et al.~\cite{dineen2017} provide a different, quadratic, formulation for \textsc{MinDS}, which requires additional variables/qubits to balance penalties where more than one vertex is selected in the dominating set.

%% file: sections/min-max-matching.tex
\section{Minimum Maximal Matching}
\label{sec:min-max}



Recall the problem \textit{\textsc{Minimum Maximal Matching (M$^3$)}} introduced at the end of Sec.~\ref{sec:profit-framework}, where we derive, as its SCOOP twin, the unconstrained problem  \textit{\textsc{Maximum Profitable Edge Set (MaxPES)}}. 






The profit formulation \textsc{MaxPES} of \textit{\textsc{M$^3$}} is derived by relaxing the two constraints of  $\textsc{M}^3$: we neither require the set of edges to be a matching nor maximal. 
Using the SCOOP framework, this process can by summarized as follows. 
%

\noindent \textbf{Step 1: Identify constraints} \\
The constraints that any solution $M\subseteq E$ to \textit{\textsc{M$^3$}} must satisfy to ensure feasibility is that $M$ is a matching that is maximal.

\noindent \textbf{Step 2: Quantify feasibility of infeasible solution} \\
Consider the \textit{extent} to which a subset \textit{covers edges}. An edge is covered by itself or is adjacent to an edge in the subset.   

\noindent \textbf{Step 3: Develop objective function $\texttt{obj}_{\text{PES}}$}\\
For a subset $\textit{PES}\subseteq E$ in $G$,\\
$
\texttt{obj}^G_{\text{PES}}(\textit{PES}) = |\{e \in E|~e~\text{is covered by}~ \textit{PES}\}| - |\textit{PES}|
$. We refer to the value produced by the objective function $\texttt{obj}^G_{\text{PES}}(\textit{PES})$ as \textit{profit}, $\mathfrak{p}_{\text{PES}}$.

\noindent \textbf{Step 4: Objective function compatibility} \\
From Theorem~\ref{thm:equiv-mmm-eds} we can derive function $f_G(.)$. 
Here, $f_G(\mathfrak{p}) = |E| - k$ and $f_G^{-1}(k) = |E| - \mathfrak{p}$. 

\noindent \textbf{Step 5: Solution enhance-ability}\\
The proof of Theorem~\ref{thm:equiv-mmm-eds} shows a two-step process  that converts any solution $\textit{PES} \subseteq E$ for a given \textit{\textsc{MaxPES}}-instance $G=(V,E)$ to a maximal matching, and therefore to a solution for \textit{\textsc{M$^3$}} while preserving solution quality.  Alg.~\ref{alg:add_edges} describes the first step of the polynomial-time post-processing procedure, i.e.~converting $\textit{PES}$ 
to a \textit{maximal} subset of edges---a subset $\textit{PES}'$ where the addition of any additional edge from $E$ would reduce the profit--while preserving solution quality. Alg.~\ref{alg:eds_to_mm} describes the second step that converts  $\textit{PES}'$ to a \textit{maximal} matching, again while preserving solution quality. 

\smallskip

\color{black}





\noindent We say that a subset that covers all edges is a \textit{maximal profitable edge set} and a subset $\textit{PES}\subseteq E$ that maximizes $\texttt{obj}^G_{\text{PES}}(\textit{PES})$ is a \textit{maximum profitable edge set}. \\

\subsubsection*{Determining minimum maximal matchings via maximum profitable edge sets}
The following theorem implies that we can derive, from any given subset of edges with profit $\mathfrak p$,  a maximal matching $\textit{MM}$ with profit at least $\mathfrak p$.

\begin{thm}\label{thm:equiv-mmm-eds} Let $G =(V,E)$  be an undirected simple graph.
$G$ has a maximal matching $\textit{MM}\subseteq E$ of size $k$ if and only if $G$ has a subset $\textit{PES}\subseteq E$ with profit ${\mathfrak p}_{\text{PES}}=|E| - k$. 
\end{thm}

\begin{proof} (Sketch)
\item $(\Rightarrow)$ 
Consider a maximal matching $\textit{MM} \subseteq E$ for $G$, $|\textit{MM}| = k$. $\textit{MM}$ has a profit of ${\mathfrak p}_{\text{PES}}=|E_{\text{PES}}(G,\textit{MM})| - |\textit{MM}|$, where $E_{PES}(G, \textit{MM})$ are the edges covered by $\textit{MM}$. Since  $\textit{MM}$ is a maximal matching, it covers all edges in $G$. Therefore, $|E_{\text{PES}}(G,\textit{MM})| = |E|$ and ${\mathfrak p}_{\text{PES}}=|E| - |k|$.

\item $(\Leftarrow)$ 
Let subset $\textit{PES} \subseteq E$ with ${\mathfrak p}_{\text{PES}} = |E| - k$.
For the case where $\textit{PES}$ is a maximal matching, there is nothing to prove.

If $\textit{PES}$ is not  a maximal matching and $E_{\text{PES}}(G, \textit{PES}) = E$, then $|\textit{PES}|=k$, and $\textit{PES}$ covers all edges,  which means that $\textit{PES}$ is a maximal profitable edge set.

If $\textit{PES}$ is neither a maximal matching nor a maximal profitable edge set, then $\textit{PES}$ does not cover all edges. One can obtain a maximal profitable edge set $\textit{PES}'$ of size at most $|E| - {\mathfrak p}_{\text{PES}}$ and profit at least ${\mathfrak p}_{\text{PES}}$  by covering  the remaining edges in $E\setminus E_{\text{PES}}(G, \textit{PES})$ with  post-processing Alg.~\ref{alg:add_edges}. 

\begin{algorithm}
\caption{Converting solutions: profitable edge sets to maximal profitable edge sets}
\label{alg:add_edges}
\begin{algorithmic}

\REQUIRE Graph $G$, solution $\textit{PES}$
\STATE Let $\textit{PES}'= \textit{PES}$
\STATE Let $E_{\text{unc}}$ be the set of uncovered edges in $G$
\FOR{each edge $(u,v) \in E_{\text{unc}}$}
    \IF{$N_e(u,v) \not\subseteq \textit{PES}$ and $(u, v) \notin \textit{PES}$}
        \STATE Add $(u, v)$ to $\textit{PES}'$
    \ENDIF
\ENDFOR
\end{algorithmic}
\end{algorithm}




Finally, if the maximal profitable edges set $\textit{PES}'$ is not a matching, it can be converted in polynomial time into a maximal matching $\textit{MM}$  with $|MM| \leq |PES'|$ (Alg.~\ref{alg:eds_to_mm}). 

\begin{algorithm}
\caption{Converting a maximal profitable edge set into a maximal matching; adapted from~\cite{yannakakis1980edge}}
\label{alg:eds_to_mm}
\begin{algorithmic}
\REQUIRE Graph $G$, maximal profitable edge set $\textit{PES}'$
\FOR{each pair of adjacent edges $(u,v), (v,w) \in \textit{PES}'$}
\IF{$\textit{PES}' \setminus (u,v)$ covers all edges}
    \STATE{Remove $(u,v)$ from $\textit{PES}'$ and break}
\ENDIF    
\IF{$\textit{PES}' \setminus (v,w)$ covers all edges}
    \STATE{Remove $(v,w)$ from $\textit{PES}'$ and break}
\ELSE
\STATE Let $S$ be the set of edges incident to $w$
\STATE Pick an edge $(w, z) \in S$, $z\not=v$, such that $(w,z)$ is covered only by $(v, w)$
\STATE $\textit{PES}' := (\textit{PES}' \setminus (v, w)) \cup (w, z)$
\ENDIF
\ENDFOR
\end{algorithmic}
\end{algorithm}




\end{proof}
\color{black}
\noindent Theorem~\ref{thm:equiv-mmm-eds} implies that, given a maximum profitable edge set $\textit{PES}$ for a graph $G$, we can obtain a minimum maximal matching for $G$ using the procedure given in the proof.  
Generally, to convert any solution to instances for \textit{\textsc{MaxPES}} to maximal matchings for the same instance without sacrificing the solution quality, we apply the 2-step procedure. 

We would like to point out a connection between the problems investigated in the section with to two other famous NP-hard problems. In 1980, Yannakakis and Garvil~\cite{yannakakis1980edge} pointed out the transformational equivalence between the problems \textit{\textsc{M}}$^3$ and \textit{\textsc{Minimum Edge Dominating Set (MinEDS)}}, the problem that for an undirected graph  seeks a smallest subset of edges that covers every edge in the graph.   When following the SCOOP framework one can see that the constrained problem M$^3$ and Minimum Edge Dominating Set share unconstrained SCOOP twin \textit{\textsc{MaxPES}}. However, to convert solutions for \textit{\textsc{MaxPES}} into solutions for \textit{\textsc{MinEDS}}, Alg.~\ref{alg:add_edges} is sufficient. Another problem investigated in this paper is \textit{\textsc{Minimum Independent Edge Dominating Set (MinIEDS)}}. An \textit{independent edge dominating} set is a dominating set that is also a matching. While there is no difference between a minimum independent edge dominating set and a minimum maximal matching, not every edge dominating set is also a maximal matching. Similarly, not every maximal matching is also an edge dominating set. \textit{\textsc{MaxPES}}, however, again serves a SCOOP twin for \textit{\textsc{MinIEDS}} as any  subset of edges of profit $\mathfrak p$ can be converted into a matching of profit at least $\mathfrak p$. 




\subsection*{Cost Hamiltonians for \textit{\textsc{MaxPES}} and \textit{\textsc{M$^3$}}} 
We describe the cost Hamiltonians $\hat{H}_{\text{M}^3}$ and $\hat{H}_{\text{PES}}$ for  \textit{\textsc{M$^3$}} and \textit{\textsc{PES}}, respectively below. 



Let $\vec{x} = (x_1, x_2, \dots, x_{|E|})$ be a vector of binary variables, where $x_i \in \{0, 1\}$ for each $i \in E$.

\vspace*{-5pt}

\begin{align*}
    \hat{H}_{PES}(\vec{x}) &=  \hat{H}^{PES}_{N}(\vec{x}) - \hat{H}^{PES}_{V}(\vec{x}) \\
    &=  \sum_{i \in E} \Bigg[1 - \Big((1-x_i)\prod_{j\in N_e(i)} (1-x_j)\Big)\Bigg] - \sum_{i \in E} x_i.
\end{align*}

\noindent For \textsc{MaxPES}, the goal is to maximize the unconstrained objective $\hat{H}_{PES}(\vec{x})$. The term $\hat{H}^{PES}_{N}(\vec{x})$ keeps track of all edges that are covered (the contribution of each edge that is covered is $1$ and $0$ otherwise).

In 2014, Lucas~\cite{Lucas2014} provided the QUBO formulation for \textsc{M$^3$}. We include it here for comparison with the \textsc{MaxPES} formulation.
Let $\vec{x} = (x_1, x_2, \dots, x_{|E|})$ be a vector of binary variables, where $x_e \in \{0, 1\}$ for each $e \in E$ representing whether or not edge $e$ is in the matching. The Hamiltonian $\hat{H}_{M^3}(\vec{x}, \vec{y})$ is defined as follows:
\begin{equation*}
\hat{H}_{M^3}(\vec{x}, \vec{y}) =\hat{H}^{M^3}_A(\vec{x}) + \hat{H}^{M^3}_B(\vec{y}) + \hat{H}^{M^3}_C(\vec{x}),
\end{equation*}
where:
\begin{equation*}
\hat{H}^{M^3}_A(\vec{x}) = A\sum_{v \in V} \sum_{\lbrace e_1,e_2\rbrace\subset \partial v} x_{e_1}x_{e_2}
\end{equation*}
penalizes having two matched edges incident to the same vertex $v$, where $\partial v$ is the set of edges incident to $v$. This term enforces the matching constraint.

Auxiliary binary variables, $\vec{y}$, are defined such that $y_v = \sum_{e\in\partial v} x_e$ (only valid for states with $H_A = 0$) which indicates if a vertex $v$ has a matched edge. Then,
\begin{equation*}
\hat{H}^{M^3}_B(\vec{y}) = B\sum_{e=(u,v) \in E} (1-y_u)(1-y_v) 
\end{equation*}

\noindent penalizes states where an edge $(u,v)$ could be added to the matching (i.e., both $y_u = 0$ and $y_v = 0$) without violating the matching constraint, thus enforcing maximality. Finally, since
\begin{equation*}
\hat{H}^{M^3}_C(\vec{x}) = C\sum_{e \in E} x_e
\end{equation*}
counts the number of matched edges, the optimal value of the total Hamiltonian corresponds to a minimum maximal matching. Lucas highlights the absence of prior work on \textsc{M}$^3$; our current analysis confirms that it remains unaddressed in the literature from the perspective of QAOA or QA.

%% file: sections/set-cover.tex
\section{Minimum Set Cover}
\label{sec:setCover}
In the previous sections, we applied our SCOOP framework to two different graph problems. Now, we apply the SCOOP framework to a different type of constrained COP, {\textsc{Minimum Set Cover}.

Given a universe ${U}$ consisting of $n$ elements and a family $\cal C$ consisting of $m$ subsets  of $U$ whose union equals $U$, a \textit{set cover} is a subfamily ${\cal C'}\subseteq \cal C$ where the union of elements in $\cal C'$ equals ${U}$.
We begin by defining the problem and before applying our SCOOP framework to derive and solve the problem via its SCOOP twin.

We define the optimization version of the problem \textit{\textsc{Minimum Set Cover (MinSC)}} for a universe ${U}$, a family of subsets ${\cal C}\subseteq {U}$, and a subfamily ${\cal C'}\subseteq {\cal C}$ as follows:
\begin{align*}
    & \text{Minimize:~} |{\cal C'}|\\
    & \text{Subject to:~} \bigcup\limits_{C_i \in {\cal C'}} C_i = {U}
\end{align*}
\noindent \textbf{Step 1: Identify constraints} \\
The constraints that any solution ${\cal C'}\subseteq {\cal C}$ to \textit{\textsc{MinSC}} must satisfy to ensure feasibility is that ${\cal C'}$ is a set cover for $U$, i.e., $\bigcup\limits_{C_i \in {\cal C'}} C_i = {U}$.

\noindent \textbf{Step 2: Quantify feasibility of infeasible solution} \\
Consider the \textit{extent} to which a subfamily of ${\cal C}$ \textit{covers the elements of } $U$.

\noindent \textbf{Step 3: Develop objective function $\texttt{obj}^{(U,{\cal C})}_{\text{PSC}}$}\\
For a subfamily ${\cal PSC}\subseteq {\cal C}$,\\

$\texttt{obj}^{(U, {\cal C})}_{\text{PSC}}(\mathcal{PSC}) = |\bigcup\limits_{C_i \in {\cal PSC}} C_i| - |\mathcal{PSC}|$. 
We refer to the value produced by the objective function $\mathfrak{p}_{\text{PSC}} = \texttt{obj}^{(U, {\cal C})}_{\text{PSC}}(\mathcal{PSC})$ as \textit{profit}.

\noindent \textbf{Step 4: Objective function compatibility} \\
From Theorem~\ref{thm:sc} we can derive function $f_U(.)$. 
Here, $f_U(\mathfrak{p}) = |U| - k$ and $f_U^{-1}(k) = |U| - \mathfrak{p}$. 

\color{black}

    Before moving on to Step 5, we formulate the problem {\sc Maximum Profitable Set Coverage} or \textit{\textsc{MaxPSC}}  for a universe  $U$, a family of subsets $\cal C$ and a subfamily $\mathcal{PSC} \subseteq \cal U$,
\begin{align*}
& \text{Maximize:} \text{~profit~} \mathfrak p_{\text{PSC}}, \text{where~} 
{\mathfrak p}_{\text{PSC}} = \texttt{obj}^{(U, {\cal C})}_{\text{PSC}}( \mathcal{PSC})
 \end{align*}


\noindent \textbf{Step 5: Solution enhance-ability}\\
The proof of Theorem~\ref{thm:sc} shows a process  that converts any solution $\mathcal{PSC} \subseteq \cal C$ for a given \textit{\textsc{MaxPSC}}-instance $(U,{\cal C})$ into a set cover while preserving solution quality.  


\begin{thm}\label{thm:sc}
For any universe $U$ and family $\cal C$ of subsets of $U$  whose union equals $U$, 
$U$ has a set cover of size $k$ if and only if there is a subfamily $\mathcal{PSC} \subseteq {\cal C}$  
with profit $\mathfrak{p}_{\text{PSC}} = |{U}| - k$.
\label{setCoverToProfitCover}
\end{thm}
\begin{proof} (Sketch) Given:  $U$,  ${\cal C}$  with $\bigcup\limits_{C_i \in {\cal C}} C_i= U$.
    \item $(\Rightarrow)$ If there is a   set cover  ${\cal C'} \subseteq {\cal C}$ with $|{\cal C'}| = k$,
    then $\bigcup\limits_{C_i \in {\cal C'}} C_i = {U}$. Thus, the profit of ${\cal C'}$ is $\mathfrak{p}_{\text{PSC}} = |\bigcup\limits_{C_i \in {\cal C'}} C_i| - k = |{U}| - k$.
    \item $(\Leftarrow)$ If there is a subfamily $\mathcal{PSC}\subseteq \cal C$  with profit $\mathfrak{p}_{\text{PSC}}$, then  $|\bigcup\limits_{C_i \in {\mathcal {PSC}}} C_i| - |\mathcal {PSC}| \geq \mathfrak{p}_{\text{PSC}}$.
    If $\mathcal {PSC}$ is a  set cover, then $\bigcup\limits_{C_i \in \mathcal {PSC}} C_i = {U}$. Thus, $|U| - |\mathcal {PSC}| \geq \mathfrak{p}_{\text{PSC}}$, yielding $|{U}| - \mathfrak{p}_{\text{PSC}} \geq |\mathcal {PSC}|$. 
 If $\mathcal {PSC}$ is not a minimum set cover, then there exists $ x\in {U}$ such that for all $C_i \subseteq \mathcal {PSC}$, $x \not\in C_i$; thus, $\bigcup\limits_{C_i \in \mathcal {PSC}} C_i \neq {U}$ and $|\bigcup\limits_{C_i \in C'} C_i| < |{U}|$. In this case we can obtain a set cover of size at most $|U| - |\mathcal {PSC}|$ by applying Algorithm~\ref{alg:add_elements}.
\end{proof}

The proof of Theorem \ref{thm:sc}, given a maximum profitable set coverage $\mathcal{PSC}$  for a universe ${U}$ and family $\cal C$ of subsets of $U$ such that $\bigcup\limits_{C_i \in {\cal C}} C_i= {U}$, tells us that we can  determine in polynomial time a minimum set cover for ${U}$ using the procedure outlined in the proof above.

\begin{algorithm}
\caption{Converting solutions: \textit{\textsc{MaxPSC}}  to \textit{\textsc{MinSC}}}
\label{alg:add_elements}
\begin{algorithmic}

\REQUIRE Universe $U$, family of subsets $\mathcal{C}$, Solution $\mathcal{PSC}$
\STATE Let $U_{\text{unc}}$ be the set of uncovered elements in $U$
\FOR{each element $x \in U_{\text{unc}}$}
   \FOR{each subfamily $C_i \in \mathcal{C}\setminus \mathcal{PSC}$}
        \IF{$x \in {C_i}$}
            \STATE Add $C_i$ to $\textit{PSC}$
        \ENDIF
    \ENDFOR
\ENDFOR
\end{algorithmic}
\end{algorithm}

\vspace*{-5pt}
\subsection*{Cost Hamiltonians for \textit{\textsc{MinSC}} and \textit{\textsc{MaxPSC}}}

\begin{figure*}[!bthp]
\centering

    \includegraphics[width=1\textwidth]{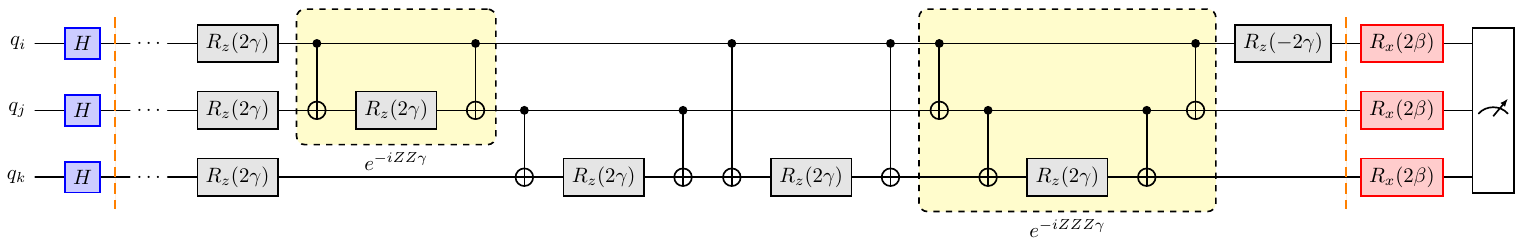}
  \caption{QAOA circuit for \textsc{MaxPD}, with $p=1$ for toy graph $G=(V, E)$, $V=\{i, j, k\}$ and $E=\{(i, j), (j,k)\}$. The circuit shows the terms and interactions relevant to $q_i$. Yellow boxes indicate gates used for quadratic ($ZZ$) and cubic ($ZZZ$) terms.} 
    \label{fig:qaoa-ciruit}
\end{figure*}

Lucas~\cite{Lucas2014} defines a QUBO for \textsc{MinSC} using the following binary variables:

\begin{itemize}
    \item $y_i = \left\{\begin{array}{ll} 
            1 & C_i \in {\cal C}' \\
            0 & C_i \not\in {\cal C}'
            \end{array}\right.$
    \item $x_{\alpha,m} = \left\{\begin{array}{ll} 
            1 & \textnormal{Number of } C_i\textnormal{'s}\in {\cal C} \textnormal{ with } \alpha\in C_i \textnormal{ is } m\geq 1\\
            0 & \textnormal{otherwise}
            \end{array}\right.$
\end{itemize}

\vspace*{-3pt}

The QUBO for \textsc{MinSC} with penalties $A$ and $B$ are then defined on  binary variables $\vec{x}$ and $\vec{y}$ as follows: 
\begin{align*}
    \hat{H}_{SC}(\vec{x},\vec{y}) &= \hat{H}^{SC}_E(\vec{x},\vec{y}) + \hat{H}^{SC}_S(\vec{y}) 
\end{align*}
where
\vspace*{-12pt}
\[\hat{H}^{SC}_E(\vec{x},\vec{y}) = A\sum_{\alpha = 1}^n \left( 1 - \sum_{m=1}^Nx_{\alpha, m} \right)^2 + \] \[A\sum_{\alpha = 1}^n \left( \sum_{m=1}^Nmx_{\alpha, m} - \sum_{i:\alpha\in C_i}y_i\right)^2\]

\noindent and 
\vspace*{-5pt}
\[\hat{H}^{SC}_S(\vec{y}) = B\sum_{i=1}^N{y_i}\]
\vspace*{-13pt}

For \textsc{MinSC}, the QUBO formulation uses \textit{penalties} to enforce the constraint that each element of the universe $U$ is covered. The penalty also activates when the variable $x_{\alpha, m}$ is counted more than once, which is necessary because we aim to count coverage only a single time (and penalize otherwise) even if the element appears in multiple subsets. 

In contrast, we formulate \textsc{MaxPSC} as a HUBO  that naturally captures requirement that each element is counted only once even if it appears in multiple subsets. without the need for penalty terms.
To describe \textsc{MaxPSC} as unconstrained binary optimization problem, we define the a HUBO. We begin with defining the  binary variables $y_i$ expressing whether or not subset $C_i$ is selected into the cover, and $x_{\alpha, i}$ that tells us whether or not element $\alpha\in U\cap C_i$:
\begin{itemize}
    \item $y_i = \left\{\begin{array}{ll} 
            1 & C_i \in {\cal C}' \\
            0 & C_i \not\in {\cal C}'
            \end{array}\right.$
    \item $x_{\alpha, i} = \left\{\begin{array}{ll} 
            1 & \alpha \in U \text{ and } \alpha \in C_i \text{ for } C_i \in {\cal C} \\
            0 & \text{otherwise}
            \end{array}\right.$
\end{itemize}
\vspace*{-3pt}
The total cost to maximize for \textsc{MaxPSC} is
\begin{align*}
\label{eq:1}
    H_{PSC}(\vec{x},\vec{y}) &= H^{PSC}_{E}(\vec{x}, \vec{y}) - H^{PSC}_{S}(\vec{y}) \\
    &= \sum_{\alpha=1}^{|U|}\left(1 - \prod_{i=1}^{|{\cal C}|}(1-x_{\alpha, i}y_i) \right) - \sum_{i=1}^{|{\cal C}|}y_i
\end{align*} 

 where $H^{PSC}_{S}(\vec{y})$ corresponds to the size of ${\cal C}'$, and $H^{PSC}_{E}(\vec{x}, \vec{y})$ to the number of elements of $U$ that are covered by the chosen subsets.  

%% file: sections/experimental-setup.tex
\vspace{-3pt}
\section{Experimental Setup}
\label{sec:methodology}

\begin{figure*}[!bthp]
\centering
  \begin{tabular}{@{}ccc@{}}
    \includegraphics[width=0.32\textwidth]{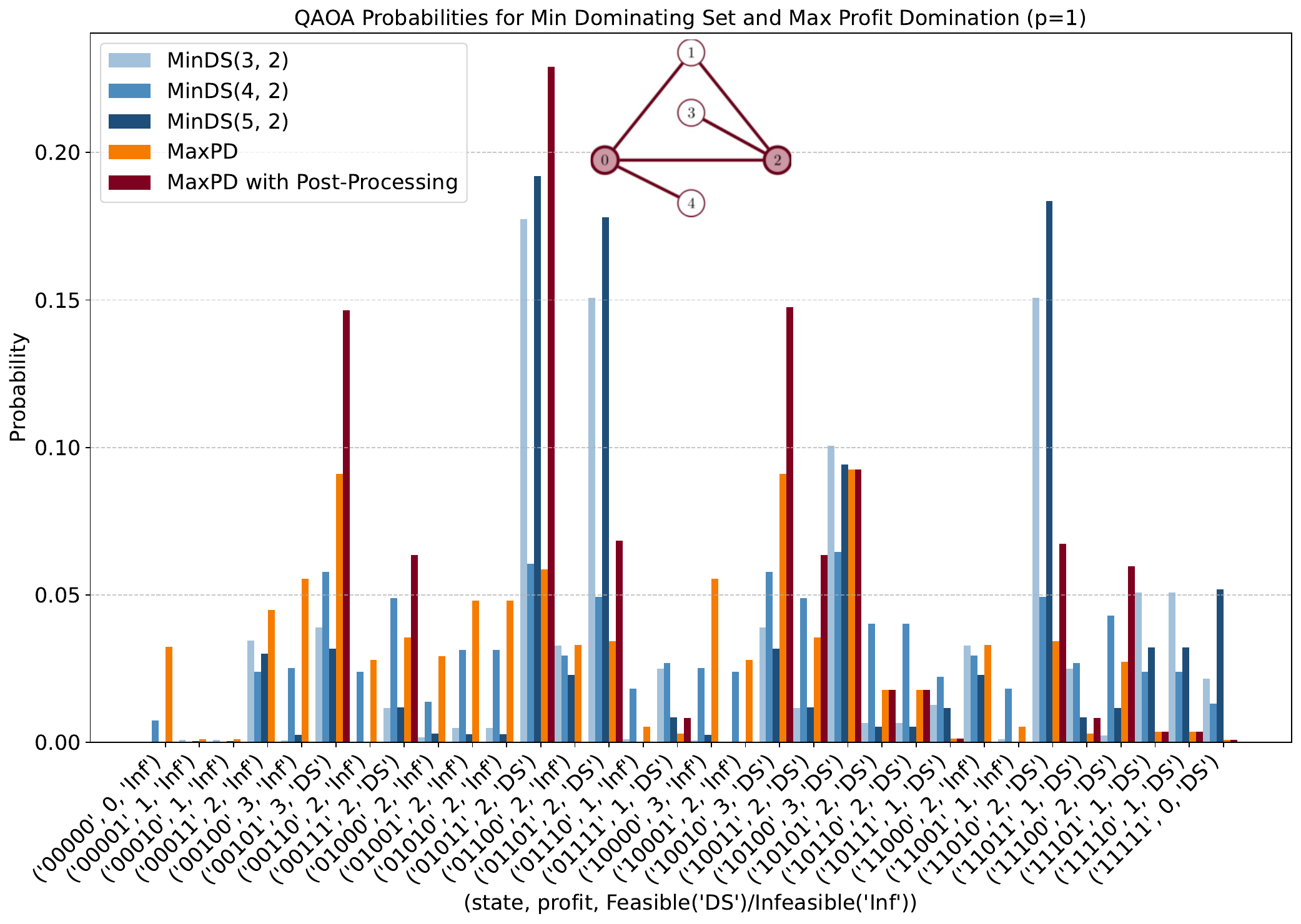} &
  
    \includegraphics[width=0.32\textwidth]{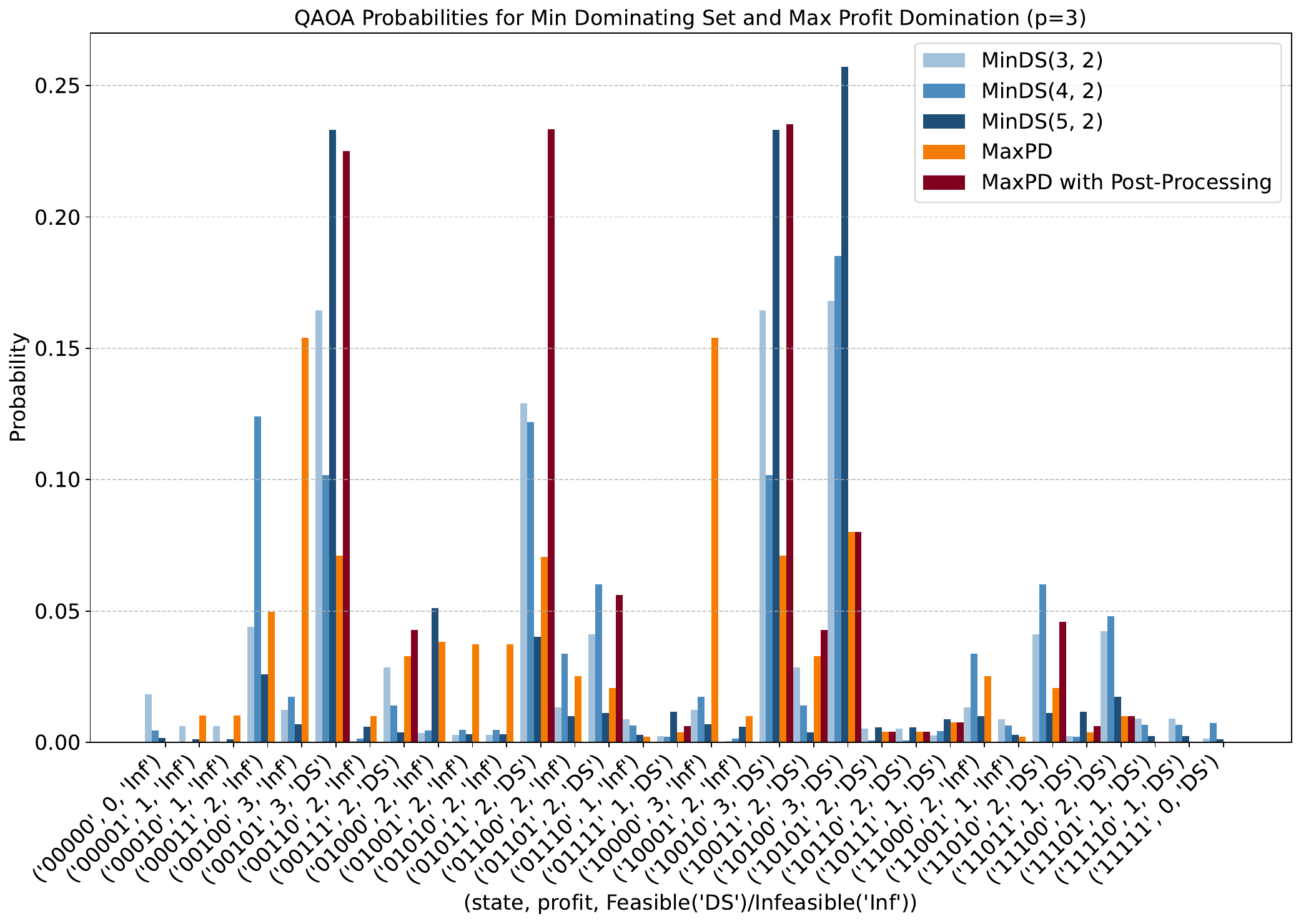} &
 
    \includegraphics[width=0.32\textwidth]{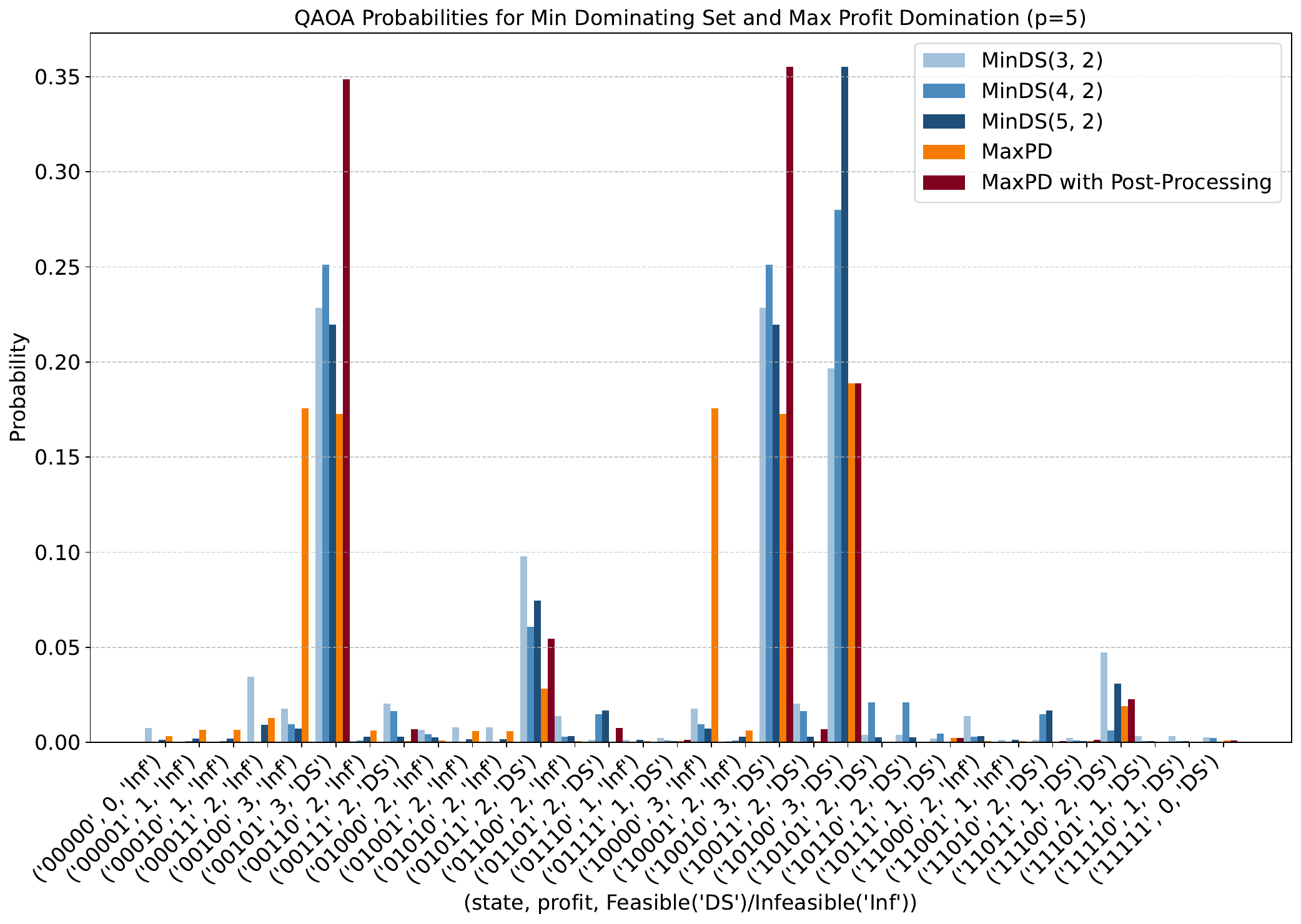}
  \end{tabular}

  \caption{Probabilities for a sample graph (shown in inset of leftmost sub-figure) containing 5 vertices with $p\in \{1,3,5\}$ layers. \textit{\textsc{MinDS}} is run for three different penalty parameters and is shown in blue. For example, Dominating Set(3, 2) refers to penalties $A=3, B=2$. Results of \textit{\textsc{MaxPD}} are shown in orange and the post-processed results of \textit{\textsc{MaxPD}} are shown in burgundy.}
    \label{fig:ds-pp-probs}
\end{figure*}

We demonstrate the feasibility and effectiveness of our approach using the SCOOP twins for constrained COPs \textit{\textsc{MinDS}} and\textit{ \textsc{M$^3$}}. While implementation efforts on \textsc{MinDS} have been limited~\cite{guerrero2020solving, dineen2017}, to the best of our knowledge, \textit{\textsc{M$^3$}} has not yet been implemented on gate-based quantum machines.
 This section details the experimental setup employed to evaluate the performance of QAOA using the SCOOP framework to compare with the standard penalty-based QAOA approach for solving Higher-Order Unconstrained Binary Optimization (HUBO) problems.

\subsubsection*{Xanadu PennyLane}

All simulations and quantum circuit constructions were implemented using PennyLane, an open-source Python framework specialized for quantum differentiable programming. For this study, the simulations were performed using PennyLane's built-in default qubit simulator with classical optimization performed using the RMSProp (Root Mean Squared Propagation) optimizer over 400 steps.

\subsubsection*{Problem Instances}

To assess the scalability and effectiveness of our approach, we generated random 3-regular graphs with $n=\{6, 8, 10\}$ nodes. These graph instances were chosen with a fixed-degree to limit the number of higher-order interactions. The structure of 3-regular graphs, when applied to the dominating set problem, results in higher-order terms in the HUBO with a maximum degree of four (quartic interactions).

\subsubsection*{QAOA Setup}
We formulate all cost Hamiltonians as minimization problems, where variables in the subset evaluated are mapped to $-1$ and the variables not chosen are mapped to $1$. For each of the cost Hamiltonians, we transform binary variables to Ising ones by applying the transformation $x_i \rightarrow \frac{1 -Z_i}{2}$. We use a basic Pauli-X mixer, $\hat{H}_M = \sum_i X_i$.

Fig.~\ref{fig:qaoa-ciruit} shows a sample QAOA circuit for the problem \textsc{MaxPD} with the terms and interactions shown for one of the nodes ($i$ with $deg(i)=2$). Assessing if this node is dominated in the set requires a Hamiltonian with 8 terms, involving interactions up to three nodes simultaneously.


 \subsubsection*{Evaluation metrics}
We performed QAOA simulations with the number of layers up to 8. To validate our proposed approach, we employ the following methodology: 
\begin{enumerate}
    \item A set of random 3-regular connected graphs of varying size are generated. 
    \item The given problem is solved using QAOA on the set of random graphs using our approach (SCOOP framework), as well as the penalty-term approach. 
    \item Results from our approach and results from the penalty-term QAOA are compared against classical exact solutions for varying graph sizes. Exact solutions to \textsc{MinDS} are obtained classically by formulating the problem as an integer linear program and solving it with the PuLP library in Python.
\end{enumerate}

Results are presented for the following evaluation metrics:
\begin{enumerate}
    \item \textbf{Probabilities} (full-statevector results) for the penalty approach (with varying penalties), for the SCOOP framework (before and after post-processing). Fig.~\ref{fig:ds-pp-probs} illustrates this for the dominating set problem using a sample graph with five nodes.
    \item We compare the penalty and SCOOP frameworks using \textbf{summed probabilities of optimal and near-optimal solutions} (we include the second and third best solutions), which represent the likelihood of obtaining high-quality solutions within the problem space in Fig.~\ref{fig:ds-opt-summed-probs}. 
    \item \textbf{Approximation ratios} of  the problems encoded using the SCOOP framework are calculated by dividing the expectation value obtained by QAOA divided by the optimal objective value calculated classically. For instance, 
\begin{equation}
\label{eq:ar}
   r_{PD} = \frac{|\langle \psi(\boldsymbol{\gamma}^*, \boldsymbol{\beta}^*)|\hat{H}_{\text{PD}}|\psi(\boldsymbol{\gamma}^*, \boldsymbol{\beta}^*)\rangle|}{|\text{{MaxPD(G)}}|} 
\end{equation}
where $(\boldsymbol{\gamma}, \boldsymbol{\beta})$ denote the parameters optimized by the classical optimizer for the problem \textsc{MaxPD}. The results, averaged over 10 randomly generated graphs, are presented in Fig.~\ref{fig:approx-ratios-maxpd-maxpes}.
\end{enumerate}

%% file: sections/results.tex
\section{Results and Discussion}
\label{sec:results}

\begin{figure*}[!htbp]
    \centering
    \begin{tabular}{@{}cccc@{}}

    \includegraphics[width=0.33\textwidth]{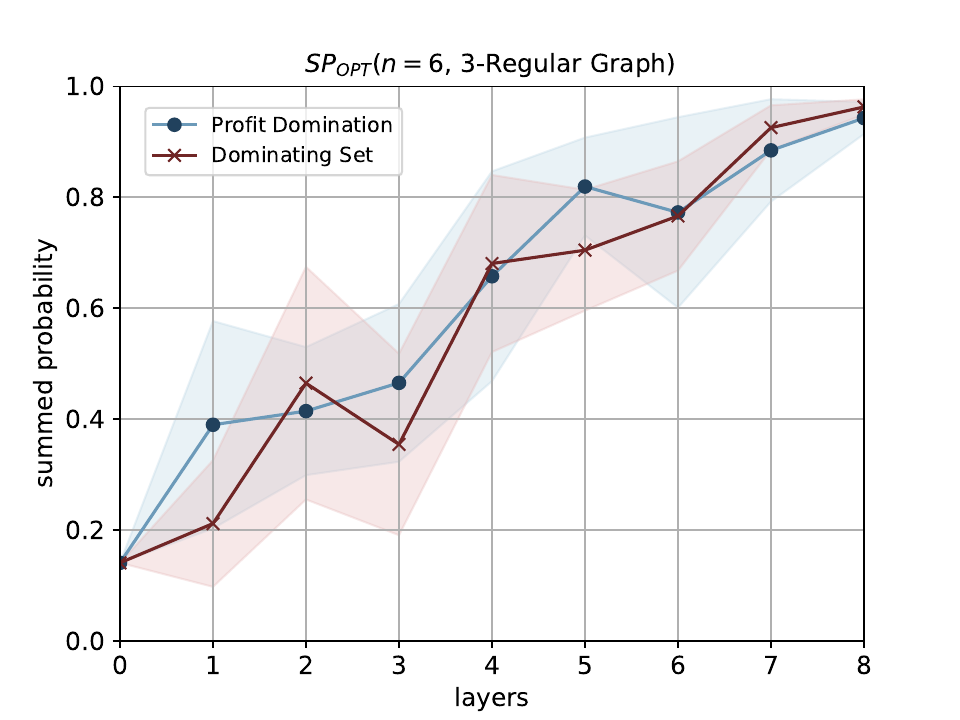} &
    \includegraphics[width=.33\textwidth]{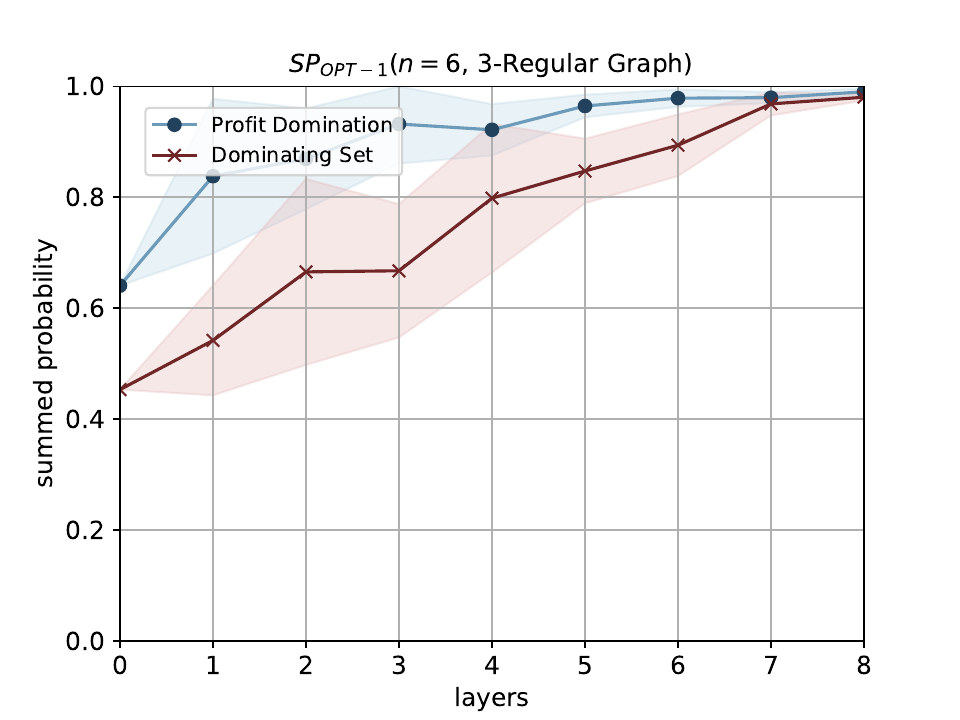} &
    \includegraphics[width=.33\textwidth]{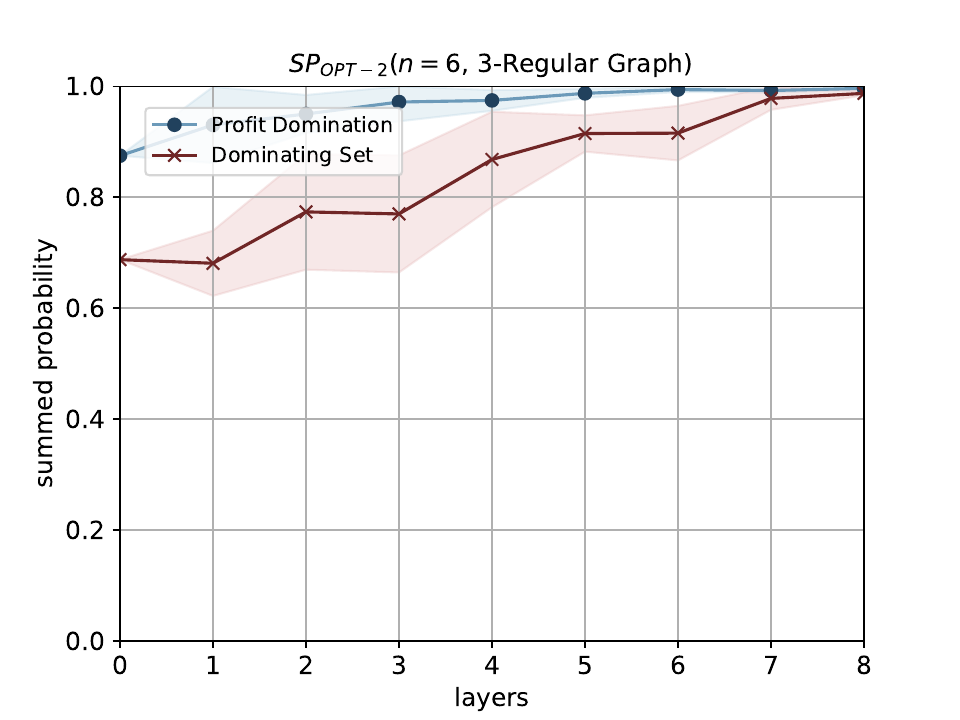} 
    \\
  \includegraphics[width=0.33\textwidth]{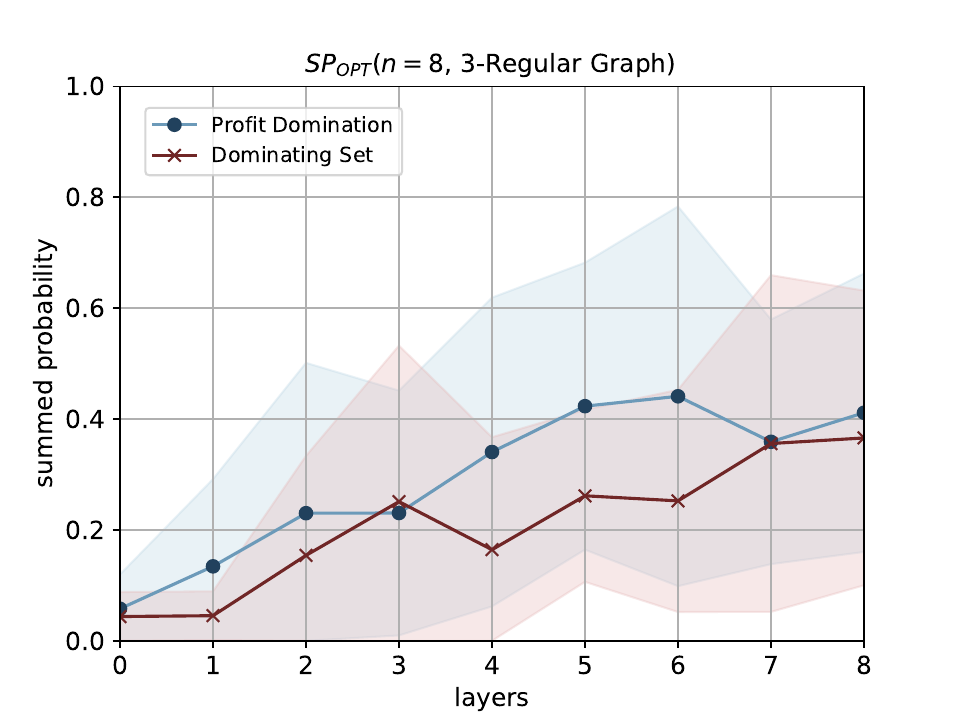} &
    \includegraphics[width=.33\textwidth]{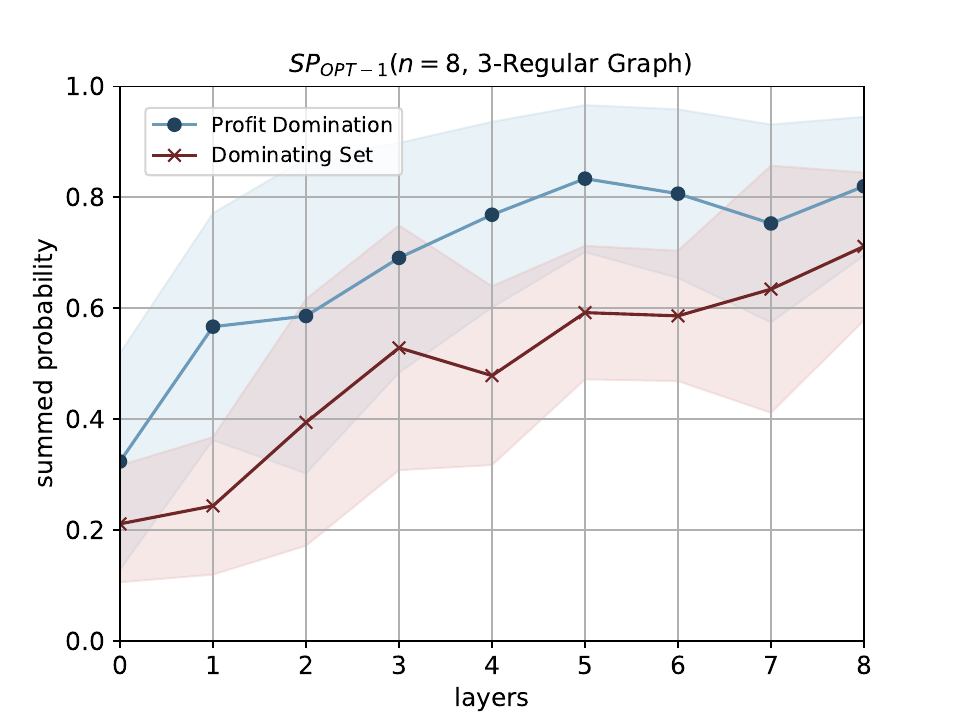} &
    \includegraphics[width=.33\textwidth]{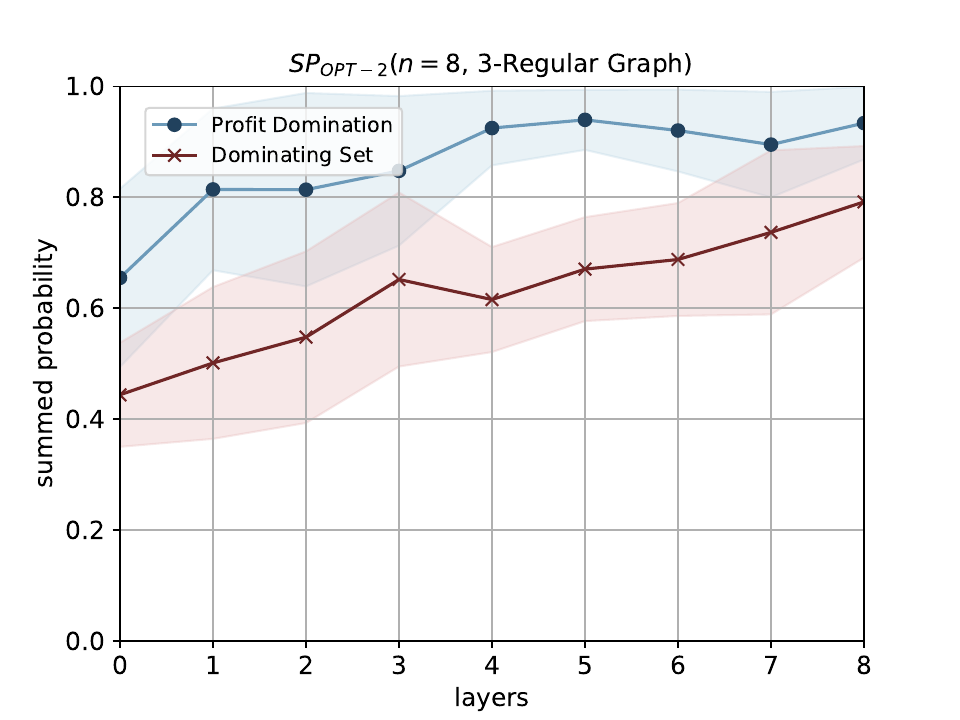} 
    \\
     \includegraphics[width=0.33\textwidth]{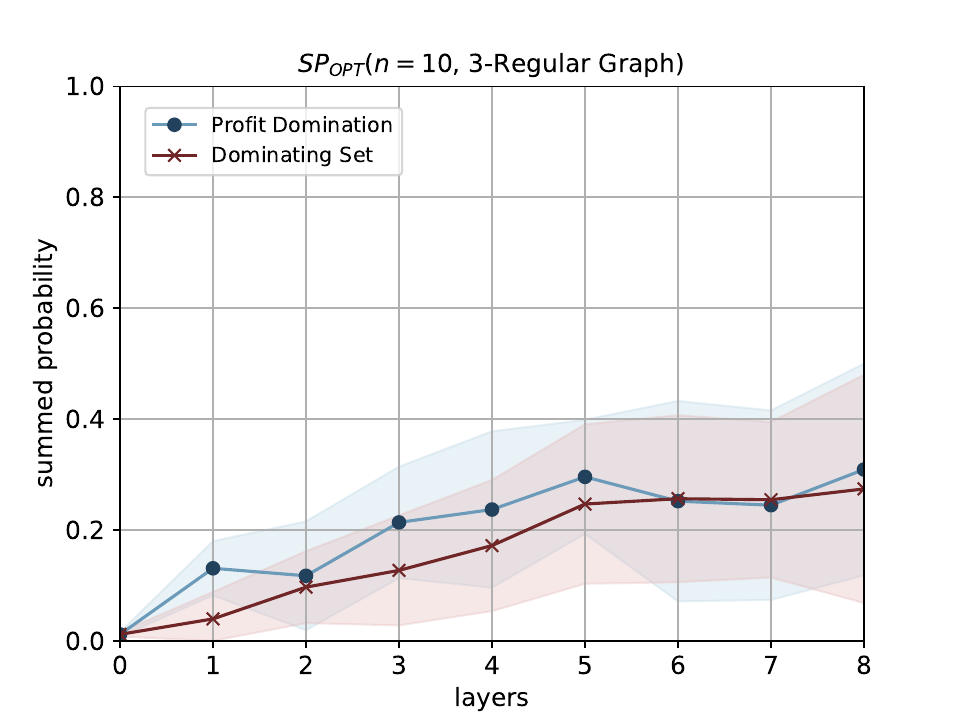} &
     \includegraphics[width=.33\textwidth]{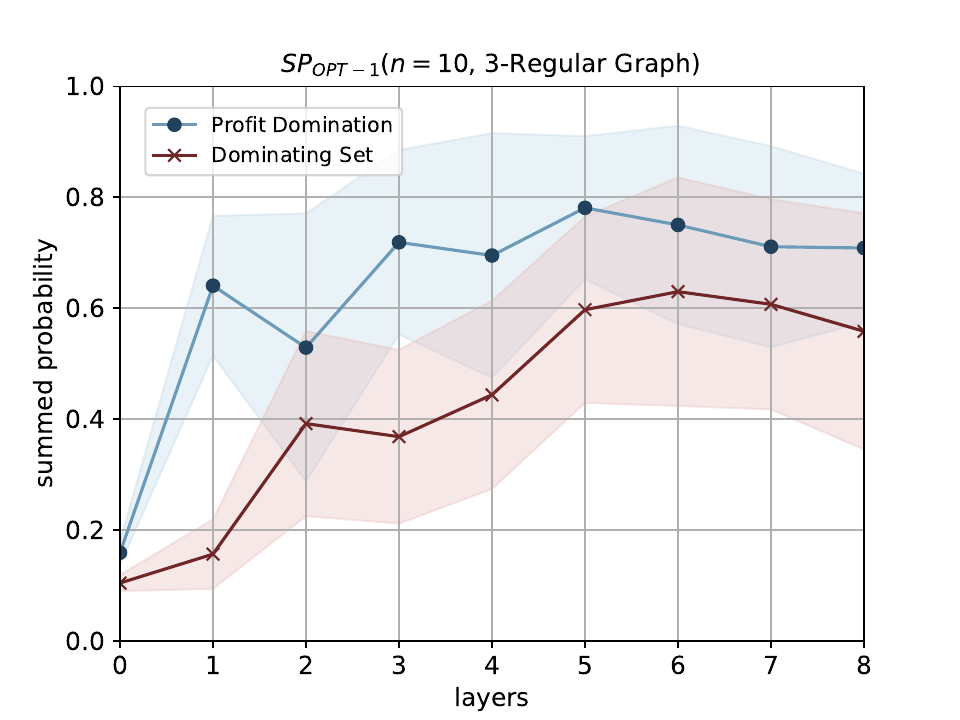} &
     \includegraphics[width=.33\textwidth]{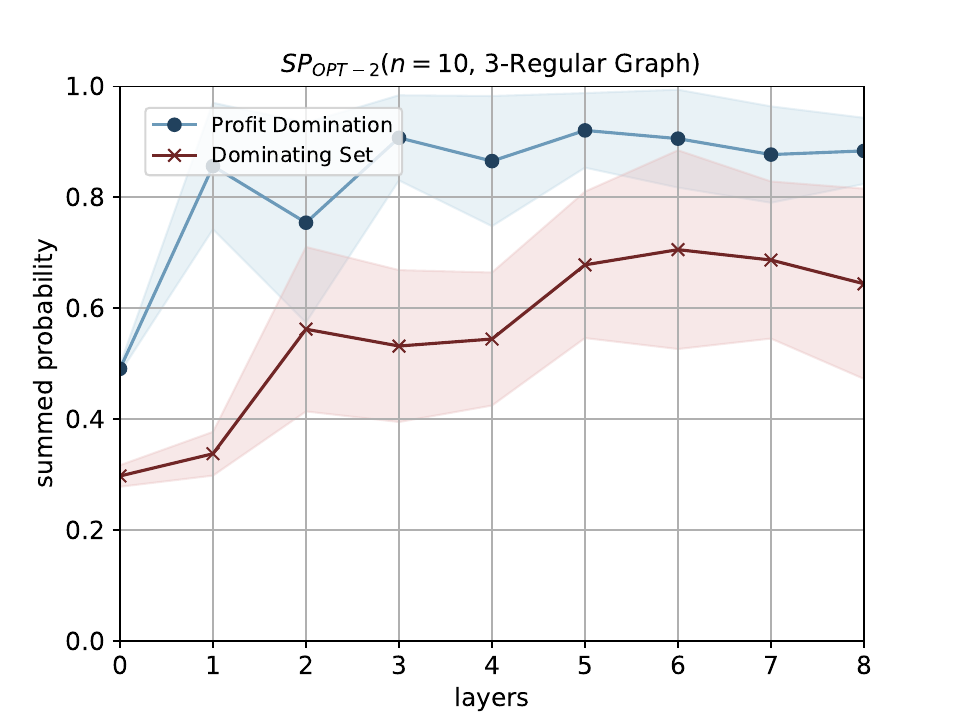} 
  \end{tabular}
    \caption{Probabilities of optimal \& near-optimal solutions obtained over 8 layers for \textit{\textsc{MinDS}} \& \textit{\textsc{MaxPD}} averaged over 10 3-regular graphs. The sub-figures in the first column show the summed probability of all the optimal solutions with $n \in \{6, 8, 10\}$. The sub-figures  in the  $2^{nd}$ column depict the summed probability of obtaining the optimal solution and the $2^{nd}$ best solution. The sub-figures  in the  $3^{rd}$ column indicates the summed probability of the optimal, $2^{nd}$ best and the $3^{rd}$ best solutions.}
    \label{fig:ds-opt-summed-probs}
\end{figure*}

\begin{figure*}[!htb]

        \centering
         \begin{tabular}{@{}cccc@{}}
         \includegraphics[width=0.45\textwidth]{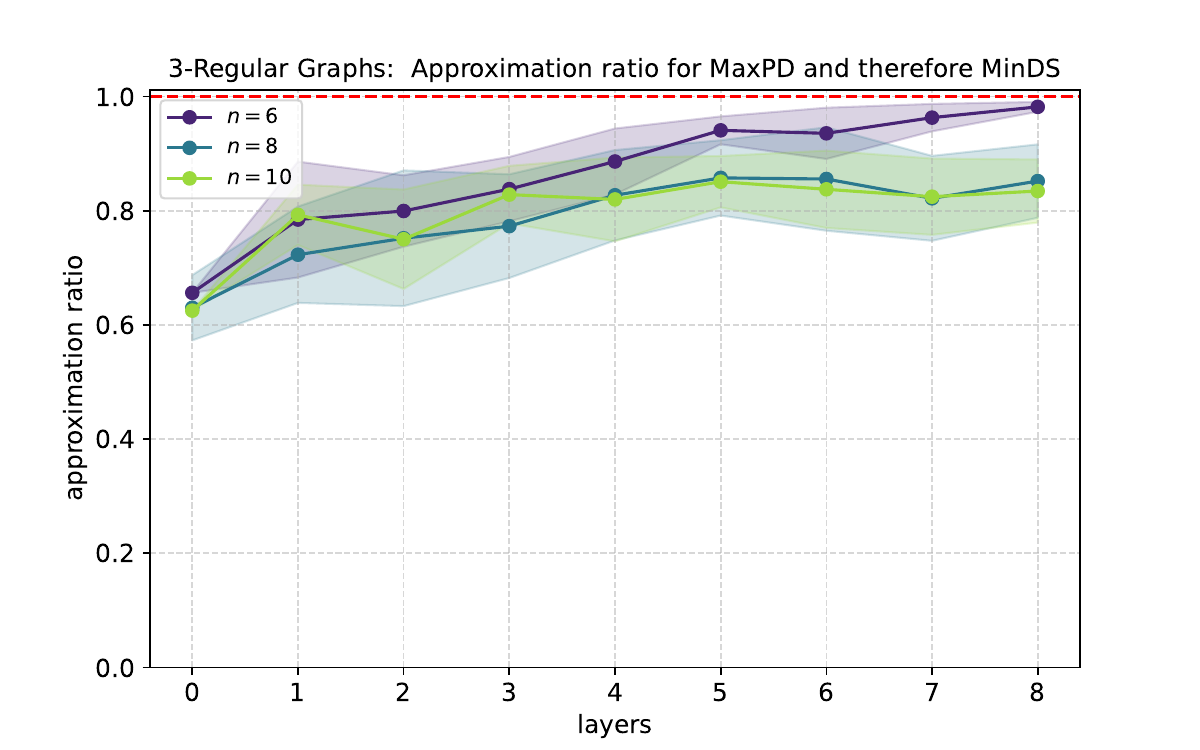} &
         \includegraphics[width=0.45\textwidth]{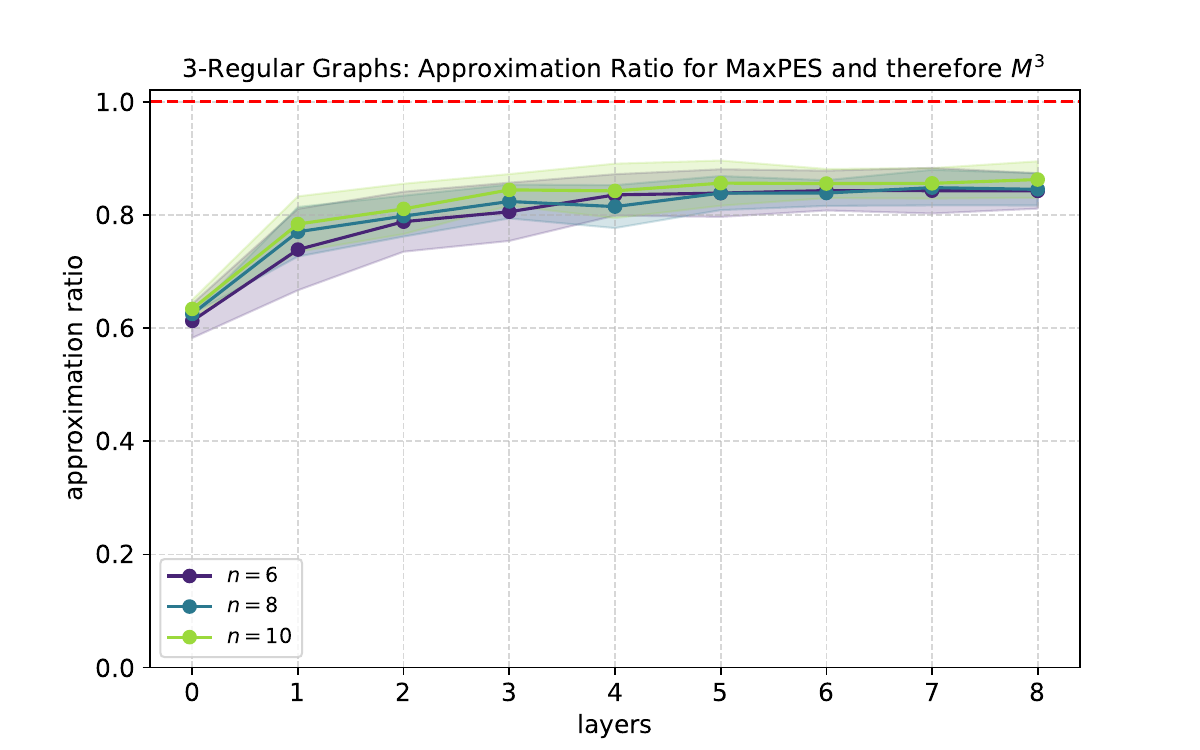}
          \end{tabular}
    \caption{Approximation ratios computed for \textit{\textsc{MaxPD}} and \textit{\textsc{MaxPES}} on PennyLane for $n\in\{6, 8, 10\}$ averaged over 10 3-regular graphs.   The profit of solutions to \textsc{MinDS} are guaranteed to be no worse than \textsc{MaxPD} (Alg.~\ref{alg:add_vertices}), and thus the approximation ratio is preserved for \textsc{MinDS}. Similarly, Algs.~\ref{alg:add_edges} and~\ref{alg:eds_to_mm} imply that this approximation ratio also applies to \textsc{M$^3$}.}
    \label{fig:approx-ratios-maxpd-maxpes}
    
\end{figure*}

In our analysis, we examine the individual probabilities of obtaining feasible solutions for the constrained problem, comparing Hamiltonians for both SCOOP twins using the standard QAOA implementation. Fig.~\ref{fig:ds-pp-probs} shows a probability distribution for \textit{\textsc{MaxPD}} and \textit{\textsc{MinDS}}  with $p\in \{1, 3, 5\}$ layers for the five node graph shown in the inset of Fig.~\ref{fig:ds-pp-probs}, run on PennyLane's standard analytical simulator (\texttt{default.qubit}). The classical optimizer used is the Root Mean Squared Propagation (RMSProp). The quantum-classical loop is run 400 times to obtain the results. The bar in the burgundy color in Fig.~\ref{fig:ds-pp-probs} shows the result of post-processing profit domination results to dominating set using Alg.~\ref{alg:add_vertices}. As the number of layers increases, we can see a high probability (greater than 0.8) of obtaining the optimal result. 



Fig.~\ref{fig:ds-opt-summed-probs} shows summed optimal and near-optimal probabilities for \textsc{MaxPD} and \textsc{MinDS}. For \textsc{MinDS} we present the results with penalties set to $A=3, B=2$. The results are averaged over 10 3-regular graphs each of sizes $n\in\{6, 8, 10\}$. For a 3-regular graph, when all three neighbors of a node interact to contribute jointly to the cost, it results in terms involving four variables (the node plus its three neighbors), leading to a polynomial with a maximum degree of 4. Thus, in this case, the QAOA operates on a cost function of degree 4, rather than the more typical degree 2 in QUBOs. 


Fig.~\ref{fig:approx-ratios-maxpd-maxpes} shows the approximation ratio obtained for \textsc{MaxPD} and \textsc{MaxPES} for $n\in\{6, 8, 10\}$ on 3-regular graphs. The derived approximation ratio applies not only to the SCOOP variants but also to the constrained versions of the problem. Specifically, Theorem~\ref{thm:equiv-ds-pd} and Theorem~\ref{thm:equiv-mmm-eds} guarantee that the approximation ratio remains valid under the imposed constraints of \textsc{MinDS} and \textsc{M$^3$}.


From Figs.~\ref{fig:ds-opt-summed-probs} and~\ref{fig:approx-ratios-maxpd-maxpes}, we observe that our method achieves an average probability of   $10\%-40\%$  of sampling the best solution at $p=1$ for 3-regular graphs with up to 10 nodes. Our approach yields over $60\%$ probability of sampling near-optimal solutions at $p=1$ for $n \in \{6, 8, 10\}$—a probability that only increases with additional layers or broader solution sets. Due to the direct relationship with its constrained SCOOP twin \textsc{MinDS}, post-processing these solutions will only make the probability better (but not worse, see Theorem~\ref{thm:equiv-ds-pd}). 

In the past, for \textsc{MinDS}, comprehensive experiments have been done only on quantum annealing hardware by Dinneen et al. in 2017~\cite{dineen2017}.  \textsc{M}$^3$, described in Lucas (2014)~\cite{Lucas2014}, has not yet been tested on QA or gate-based hardware.


        


%% file: sections/conclusions.tex
\section{Conclusions and Future Work}
\label{sec:conclusions}

We introduced our novel SCOOP framework (Sec.~\ref{sec:profit-framework}) that can serve as a blueprint to solve constrained COPs using vanilla QAOA. The SCOOP framework enables the discovery of optimal and  -optimal solutions. Our framework guides the derivation of an unconstrained COP from the constrained COP. This process results in a penalty-free cost Hamiltonian that is well-suited for QAOA. Subsequently, the viable solutions obtained by QAOA are efficiently post-processed. Our approach comes with the benefits that it scales for all inputs as it avoids penalizing solutions that do not satisfy the constraints, and supports the determination of near-optimal solutions. 

As use cases to apply the SCOOP framework (Secs.~\ref{sec:ds}--\ref{sec:setCover}) we chose the NP-hard COPs \textit{\textsc{MinDS}}, \textit{\textsc{M$^3$}}, and \textit{\textsc{MinSC}}. For each, we determine its unconstrained SCOOP twin and describe its cost Hamiltonian as a HUBO. With \textit{\textsc{MaxPES}}, and \textit{\textsc{MaxPSC}} we not only introduced new unconstrained COPs but also showed them to be NP-hard. 

Our SCOOP framework not only applies to the problems reported in this paper, but also to other NP-hard constrained COPs such as the three transformational equivalent problems (\textsc{MinVC}, \textsc{MaxIS}, and \textsc{MaxCl}) featured in \cite{Angara2025}.   Table~\ref{tab:scpp_twin_mapping} gives a summary of problems that are identified to be SCOOP twins.

For both \textit{\textsc{MinDS}} and \textit{\textsc{M$^3$}}, we discuss our experimental results in terms of probabilities, summed probabilities of optimal and near-optimal solutions, and approximation ratios. These experiments were conducted using Xanadu’s PennyLane simulator on 3-regular graphs with up to ten qubits and eight layers of QAOA.





\subsubsection*{Future Work} 
As part of future work, we aim to extend our framework to a broader class of constrained optimization problems. We plan to investigate HUBO formulations within the context of QAOA by analyzing the complexity of their cost landscapes and comparing them to their quadratized counterparts. Furthermore, we intend to execute both QUBO and HUBO formulations of the unconstrained SCOOP twins on IBM's quantum hardware.
Argonne QTensor, a tensor-network simulator, enables efficient classical simulation of quantum circuits using tensor-network contractions. This is well studied for  \textit{\textsc{MaxCut}}  with QAOA~\cite{lykov2021performance}. Design principles of QTensor allow for ease of extensibility of the framework with minimal changes to the underlying QTensor architecture. In Angara et al.~\cite{Angara2025}, QTensor is used to simulate QUBOs for the unconstrained twins of constrained optimization problems on sparse and 3-regular graphs with up to 70 nodes. In its current state, only QUBOs can be formulated for tensor network simulation. We plan to extend this framework to include formulations that allow higher-order terms and their expectation value calculation on QTensor. 